\documentclass[aps,english,prx,floatfix,amsmath,superscriptaddress,tightenlines,twocolumn,nofootinbib]{revtex4-2}
\usepackage{mathtools, amssymb}
\usepackage{graphicx}
\usepackage{tikz}
\usepackage{tikz-3dplot}
\usetikzlibrary{arrows.meta}
\usetikzlibrary{calc}
\usepackage{intcalc}
\usepackage[caption=false]{subfig}
\usepackage[shortlabels]{enumitem}

\usepackage{multirow}
\usepackage{hhline}
\usepackage{algorithmicx}
\usepackage{algorithm}
\usepackage[noend]{algpseudocode}

\usepackage{soul}
\usepackage[utf8]{inputenc}
\usepackage{xcolor}
\usepackage{amsthm}
\usepackage{float}
\usepackage{comment}
\usepackage{qcircuit}

\usepackage{amsthm}
\usepackage{thmtools}
\declaretheorem[]{theorem}

\declaretheorem[numberwithin=theorem]{corollary}

\makeatletter

\newcommand{\Tr}[0]{\text{Tr}}

\newcommand{\axiomone}[2]
{
	\begin{scope}[xshift=#1cm, yshift=#2cm]
		\draw[line width=0.6pt, fill=red!10!white, opacity=0.5] (0,0) circle (0.3cm);
		\draw[line width=0.6pt] (0,0) circle (0.3cm);
		\draw[line width=0.6pt] (0,0) circle (0.15cm);
	\end{scope}
}
\newcommand{\axiomtwoangled}[4]
{
	\begin{scope}[xshift=#1cm, yshift=#2cm]
		\draw[line width=0.6pt, fill=green!10!white, opacity=0.5] (0,0) circle (0.3cm);
		\draw[line width=0.6pt] (0,0) circle (0.3cm);
		\draw[line width=0.6pt] (0,0) circle (0.15cm);
		\draw[line width=0.6pt] (#3:0.15) -- (#3:0.3);
		\draw[line width=0.6pt] (#4:0.15) -- (#4:0.3);
	\end{scope}
}

\newcommand{\coloredboxhx}[5]{
	\foreach \x in {#2, ..., #4}
	{
		\foreach \y in {#3, ..., #5}
		{
			\fill[#1] (\x-0.5*\y +0.5, \y*0.866 -0.289) -- (\x-0.5*\y +0.5, \y*0.866 +0.289) -- (\x-0.5*\y, \y*0.866+0.577) -- (\x-0.5*\y -0.5, \y*0.866 +0.289) --(\x-0.5*\y-0.5, \y*0.866 -0.289)--(\x-0.5*\y, \y*0.866 -0.577)--(\x-0.5*\y+0.5, \y*0.866 -0.289) -- cycle;
			
			\draw (\x-0.5*\y +0.5, \y*0.866 -0.289) -- (\x-0.5*\y +0.5, \y*0.866 +0.289) -- (\x-0.5*\y, \y*0.866+0.577) -- (\x-0.5*\y -0.5, \y*0.866 +0.289) --(\x-0.5*\y-0.5, \y*0.866 -0.289)--(\x-0.5*\y, \y*0.866 -0.577)--(\x-0.5*\y+0.5, \y*0.866 -0.289) -- cycle;
		}
	}
}

\newcommand{\coloredkv}[4]{
	\begin{scope}[xshift=#3 cm, xshift=-0.5*#4 cm, yshift=0.866*#4 cm, rotate=60*#2]
		\fill[#1] (-30:0.578) -- (30:0.578) -- (90:0.578) -- (150:0.578) --(210:0.578)--(270:0.578) -- cycle;
		
		\draw (-30:0.578) -- (30:0.578) -- (90:0.578) -- (150:0.578) --(210:0.578)--(270:0.578) -- cycle;
	\end{scope}	
}
\newcommand{\coloredke}[4]{
	\begin{scope}[xshift=#3 cm, xshift=-0.5*#4 cm, yshift=0.866*#4 cm, rotate=60*#2]
		\fill[#1] (-30:0.578) -- (30:0.578) -- (90:0.578) -- (150:0.578) --(210:0.578)--(270:0.578) -- cycle;
		\fill[#1,xshift=1 cm] (-30:0.578) -- (30:0.578) -- (90:0.578) -- (150:0.578) --(210:0.578)--(270:0.578) -- cycle;
		
		\draw (30:0.578) -- (90:0.578) -- (150:0.578) --(210:0.578)--(270:0.578)--(330:0.578)--++(-30:0.578)--++(30:0.578)--++(90:0.578)--++(150:0.578)--++(210:0.578)--cycle;
	\end{scope}	
}
\newcommand{\coloredkf}[4]{
	\begin{scope}[xshift=#3 cm, xshift=-0.5*#4 cm, yshift=0.866*#4 cm, rotate=60*#2]
		\fill[#1] (-30:0.578) -- (30:0.578) -- (90:0.578) -- (150:0.578) --(210:0.578)--(270:0.578) -- cycle;
		\fill[#1,xshift=1 cm] (-30:0.578) -- (30:0.578) -- (90:0.578) -- (150:0.578) --(210:0.578)--(270:0.578) -- cycle;
		\fill[#1,xshift=0.5 cm,yshift=0.866 cm] (-30:0.578) -- (30:0.578) -- (90:0.578) -- (150:0.578) --(210:0.578)--(270:0.578) -- cycle;
		
		\draw (90:0.578) -- (150:0.578) --(210:0.578)--(270:0.578)--(330:0.578)--++(-30:0.578)--++(30:0.578)--++(90:0.578)--++(150:0.578)
		--++(90:0.578)--++(150:0.578)--++(210:0.578)--++(270:0.578)--cycle;
	\end{scope}	
}

\newcommand{\diskshadded}[6]
{
	\foreach \x in {#1,...,#3}
	{
		\foreach \y in {#2,...,#4}
		{
			
			\fill[#6] (\x-0.5*\y +0.5, \y*0.866 -0.289) -- (\x-0.5*\y +0.5, \y*0.866 +0.289) -- (\x-0.5*\y, \y*0.866+0.577) -- (\x-0.5*\y -0.5, \y*0.866 +0.289) --(\x-0.5*\y-0.5, \y*0.866 -0.289)--(\x-0.5*\y, \y*0.866 -0.577)--(\x-0.5*\y+0.5, \y*0.866 -0.289) -- cycle;
			
			\draw[color=#5] (\x-0.5*\y +0.5, \y*0.866 -0.289) -- (\x-0.5*\y +0.5, \y*0.866 +0.289) -- (\x-0.5*\y, \y*0.866+0.577) -- (\x-0.5*\y -0.5, \y*0.866 +0.289) --(\x-0.5*\y-0.5, \y*0.866 -0.289)--(\x-0.5*\y, \y*0.866 -0.577)--(\x-0.5*\y+0.5, \y*0.866 -0.289) -- cycle;

		}
	}
}

\usetikzlibrary{spy}

\def\xcolor{blue!50!cyan!80!green!50!white}
\def\ycolor{blue!50!cyan!60!white}
\def\zcolor{blue!60!cyan!75!red!70!white}

\definecolor{maroon}{rgb}{144,12,63}
\definecolor{darkblue}{rgb}{27,12,144}
\definecolor{mypurple2}{RGB}{170,0,255}
\definecolor{myred}{RGB}{255, 0, 90}
\definecolor{mycyan}{RGB}{0, 191, 255}
\definecolor{KK}{RGB}{0, 180, 20}

\usepackage{hyperref}
\hypersetup{colorlinks=true,linkcolor=myred,citecolor=magenta,urlcolor=blue}
\usepackage[all]{hypcap}

\begin{document}
\author{Isaac H. Kim}
\thanks{Equal contribution.}
\affiliation{Department of Computer Science, University of California, Davis, Davis, CA, USA}
\affiliation{Centre for Engineered Quantum Systems, School of Physics, University of Sydney, Sydney, NSW, Australia}
\author{Bowen Shi}
\thanks{Equal contribution.}
\affiliation{Department of Physics, University of California at San Diego, La Jolla, CA, USA}
\author{Kohtaro Kato}
\affiliation{Center for Quantum Information and Quantum Biology, Osaka University, Osaka, Japan}
\affiliation{Graduate School of Informatics, Nagoya University, Nagoya, Japan}
\author{Victor V. Albert}
\affiliation{Joint Center for Quantum Information and Computer Science, National Institute of Standards and Technology and University of Maryland, College Park, MD, USA}
\title{Modular commutator in gapped quantum many-body systems}
\date{\today}

\begin{abstract} 
In \href{https://arxiv.org/abs/2110.06932}{arXiv:2110.06932}, we argued that the chiral central charge --- a topologically protected quantity characterizing the edge theory of a gapped (2+1)-dimensional system --- can be extracted from the bulk by using an order parameter called the modular commutator. In this paper,  we reveal general properties of the modular commutator and strengthen its relationship with the chiral central charge. First, we identify connections between the modular commutator and conditional mutual information, time reversal, and modular flow. Second, we prove, within the framework of the entanglement bootstrap program, that two topologically ordered media connected by a gapped domain wall must have the same modular commutator in their respective bulk. Third, we numerically calculate the value of the modular commutator for a bosonic lattice Laughlin state for finite sizes and extrapolate to the infinite-volume limit. The result of this extrapolation is consistent with the proposed formula up to an error of about $0.7\%$.
\end{abstract}
\maketitle

\section{Introduction}

Two-dimensional quantum many-body systems with a bulk energy gap can host a number of fascinating physical phenomena. The well-known fractional quantum Hall states~\cite{FQHE} can host anyons~\cite{Leinaas1977,Wilczek1982,Arovas1984}, topologically protected ground state degeneracy~\cite{Wen1990}, and long-range entanglement~\cite{Kitaev2006,Levin2006}. Another important aspect of these systems is the presence of chiral (\emph{i.e.}, unidirectional) gapless modes at the edge, which are intimately related to the systems' well-known quantized Hall conductance~\cite{Klitzing1980,FQHE,Thouless1982,Girvin2019}. 

Quantum Hall systems have a $U(1)$ symmetry associated with charge conservation, but protected chiral gapless edge modes can remain even in the absence of such a symmetry~\cite{Kane1997}. The relevant transport coefficient in this context is the thermal Hall conductance, which is related to the energy current at the edge. This current, at a temperature $T$ low compared to the bulk energy gap, is
\begin{equation}
    I=\frac{\pi}{12}c_- T^2, \label{eq:edge_current}
\end{equation}
where $c_-$ is the \emph{chiral central  charge}. While this formula can be derived assuming the edge is described by a conformal field theory, this formula is valid even if all the symmetries are absent~\cite{Kane1997,Kitaev2006a}.

While the chiral central charge is a property of the edge, it is also intimately related to the properties of the bulk. For instance, the energy current at the edge can be related to the energy $2$-current in the bulk, which can be computed from a microscopic Hamiltonian~\cite{Kitaev2006a,Kapustin2020}. In effective field theory approaches, chiral central charge  appears in the gravitational Chern-Simons term of the bulk action, which is responsible for the framing anomaly of the underlying system~\cite{Gromov2015}. This is responsible for the appearance of chiral central chage in the Hall viscosity on a sphere~\cite{Abanov2014,Klevtsov2015,Bradlyn2015,Golan2019}.

Moreover, the following relation, which applies to bosonic topologically ordered systems, is well-known: 
\begin{equation}
    \mathcal{D}^{-1} \sum_a d_a^2 \theta_a = e^{2\pi i c_-/8},\label{eq:bulk_edge}
\end{equation}
where  $d_a$ is the quantum dimension of the superselection sector $a$, $\theta_a$ is its topological spin, and $\mathcal{D}=\sqrt{\sum_a d_a^2}$ is the total quantum dimension~\cite{Frolich1990,Rehren1989}. The left hand side of this equation is determined completely by the low-energy excitations of the bulk whereas the right hand side is the property of the edge. One can view this equation as a manifestation of the bulk-edge correspondence~\cite{Moore1991,Read1992,Read1999,Li2008}.

It is also known that the chiral central charge leaves an imprint in the ground state entanglement. For instance, one can extract the chiral central charge from the entanglement spectrum of the ground state~\cite{Li2008}. Moreover, there are methods to compute the chiral central charge (modulo some constant, \emph{e.g.,} 8) from the ground state wave function, such as the modular $S$- and $T$-matrix~\cite{Zhang2012}, momentum polarization~\cite{Tu2013}, or the Berry curvature of the ground state wave function under adiabatic variation of the metric~\cite{Bradlyn2015}. 

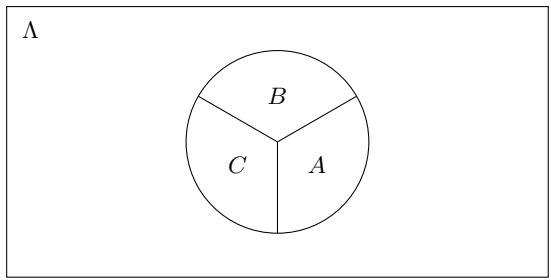
\begin{figure}[h]
    \centering
    \begin{tikzpicture}[scale=0.45]
    \draw[] (-8, -4) -- (8,-4) -- (8,4) -- (-8,4) -- cycle;
    \node[below right] (lambda) at (-7.8,3.8) {$\Lambda$};
    \draw (0,0) circle (2.7);
    \draw (0,0) -- (30:2.7);
    \draw (0,0) -- (150:2.7);
    \draw (0,0) -- (-90:2.7);
    \node at (90:1.35) {$B$};
    \node at (-150:1.35) {$C$};
    \node at (-30:1.35) {$A$};
    \end{tikzpicture}
    \caption{Partition of a disk-shaped region $ABC$ in the bulk ($\Lambda$). Each subsystem is assumed to be sufficiently large compared to the correlation length.}
    \label{fig:abc}
\end{figure}

Recently, we derived a new formula for the chiral central charge that can be obtained from a single ground state wave function~\cite{Short}. Our formula is defined in terms of the \emph{modular Hamiltonian}~\cite{Haag2012} (also known as the entanglement Hamiltonian). Let $\sigma = |\psi\rangle\!\langle \psi|$ be a many-body ground state and $\sigma_A$ be its reduced density matrix over a subsystem $A$. The modular Hamiltonian $K_A$ associated with the region $A$ is $K_A = -\ln \sigma_A$. Our formula for the chiral central charge reads:
\begin{equation}
     i\langle[K_{AB}, K_{BC}] \rangle=\frac{\pi}{3} c_-, ~\label{eq:main_result}
\end{equation}
for a set of subsystems depicted in Fig.~\ref{fig:abc}, where $\langle \ldots \rangle$ is the expectation value over the ground state. Eq.~\eqref{eq:main_result} is insensitive to continuous deformation of the subsystems, so long as they remain to partition a disk. 

There are several remarkable aspects about Eq.~\eqref{eq:main_result}. The formula depends only on a single wave function, instead of a family of wave functions~\cite{Tu2013,Bradlyn2015} or the entire ground state subspace~\cite{Zhang2012}. Moreover, the quantity that we use to extract the chiral central charge is not a known measure of entanglement, to the best of our knowledge. These facts motivate us to study the properties of the modular commutator and its relations to the ground state entanglement of gapped quantum many-body systems. We initiate these studies in this manuscript.

Our contribution can be divided into roughly two categories. First, we explore various properties of the modular commutator, including its relation to conditional mutual information~\cite{Lieb1973}, time-reversal, and modular flow~\cite{Haag2012}. These are results that hold for any many-body quantum states, as long as the underlying Hilbert space has a tensor product structure. That a single quantity has relations to these seemingly disparate subjects suggest that the modular commutator may be useful more broardly.

Second, we present further evidence for the validity of Eq.~\eqref{eq:main_result} in the form of a nontrivial consistency check for systems consisting of two different gapped phases sharing a domain wall. In such cases, one expects the two phases to admit the same chiral central charge. Using a part of the entanglement bootstrap program~\cite{SKK2019} dealing with gapped domain walls~\cite{Shi2021}, we show that such a pair of phases has the same value of the modular commutator.

Another part of our evidence comes from a numerical investigation. We compute Eq.~\eqref{eq:main_result} for wave functions associated with a topological phase with a known value of chiral central charge. Specifically, we study a bosonic lattice Laughlin state proposed in Ref.~\cite{Nielsen2012}. Due to finite-size effects, the computed value of Eq.~\eqref{eq:main_result} deviates from the predicted value with a relative error of $\sim 10\%$. We carefully analyze the scaling behavior of the chiral central charge, which agrees well with an exponenially decaying function in the length scale of the subsystems used in Eq.~\eqref{eq:main_result}. By fitting the numerically obtained data to this ansatz and extrapolating to the thermodynamic limit, we obtain a relative error of $0.7\%$. 

The rest of the paper is organized as follows. In Section~\ref{sec:modular_commutator}, we discuss general properties of the commutator of modular Hamiltonians. In Section~\ref{sec:area_law}, we discuss the entanglement property of the gapped ground state wave functions, setting up the discussion in the rest of the paper. In Section~\ref{sec:locality_modular_operator}, we prove that the modular Hamiltonian in gapped systems is a sum of local terms. In Section~\ref{sec:modular_current}, we define the modular current and study its properties. We also review an argument relating edge energy current and the modular current calculated from ground state wave function~\cite{Short}. In Section~\ref{sec:gapped_domain_wall}, we prove that the modular commutator calculated from both sides of a gapped domain wall are identical, further supporting our formula for chiral central charge. In Section~\ref{sec:numerical_evidence} we discuss our numerical experiment. We end with a discussion in Section~\ref{sec:discussion}.

\section{Modular commutator}
\label{sec:modular_commutator}

In this paper, we study the expectation value of the commutator of modular Hamiltonians, which we refer to as the \emph{modular commutator}:
\begin{equation}
    J(A,B,C)_{\rho} := i\text{Tr}\left(\rho_{ABC}[K_{AB}(\rho), K_{BC}(\rho)] \right).
    \label{eq:mod_comm_definition}
\end{equation}
Here, the modular Hamiltonian of $\rho$ on subsystem $X$ is denoted as $K_{X}(\rho)= -\ln \rho_{X}$. For simplicity, we have suppressed the tensor product with the identity operator. For instance, $K_{AB}(\rho)$ in Eq.~\eqref{eq:mod_comm_definition} means $K_{AB}(\rho)\otimes 1_C$; we shall drop the label $\rho$ whenever this is obvious from the context.
We shall assume throughout this paper that the tripartite state $\rho_{ABC}$ lives on a finite-dimensional Hilbert space with a tensor product structure $\mathcal{H}_{ABC}=\mathcal{H}_{A}\otimes \mathcal{H}_{B} \otimes \mathcal{H}_{C}$. Here we discuss general properties of the modular commutator, deriving how it behaves under time-reversal, and making contact with quantum Markov chains~\cite{Petz2003} and modular flow~\cite{Haag2012}.

Let us first begin with three elementary properties. First, the modular commutator is bounded, assuming that the Hilbert space associated with $A, B,$ and $C$ are all finite-dimensional. 
(This is a somewhat nontrivial statement because, for density matrices with zero eigenvalues, the modular Hamiltonian can diverge due to the logarithm.)
To see why, it is helpful to rewrite the modular commutator as follows:
\begin{equation}
    J(A,B,C)_{\rho} =- 2\text{Im}(\langle \widetilde{\psi}_{AB} |\widetilde{\psi}_{BC}\rangle),
\end{equation}
where
\begin{subequations}
	\begin{align}
		|\widetilde{\psi}_X\rangle &= K_{X}  \vert \psi_{ABCD}\rangle\\
		&= - \sum_i \sqrt{\lambda_i} \ln \lambda_i |i_X\rangle \otimes |i_{\bar{X}}\rangle.
	\end{align} 
\end{subequations}
Here $\bar{X}$ is the complement of $X$ on systems $ABCD$, where $D$ is a finite-dimensional Hilbert space which purifies $\rho_{ABC}$ into $\vert \psi_{ABCD} \rangle$; $\{\lambda_i\}$ is the set of eigenvalues of the reduced density matrix over $X$, and $\{|i_X \rangle \}$ and $\{|i_{\bar{X}}\rangle \}$ are two orthonormal sets of vectors. Note that the function $\sqrt{\lambda} \ln \lambda \to 0$ as $\lambda \to 0$. Therefore, the norm of $|\widetilde{\psi}_X\rangle$ is bounded for any $X$. In particular, we see that $J(A,B,C)_{\rho}$ is bounded as well.

Secondly, because $[K_{AB}, K_{BC}]$ is an antihermitian operator, $J(A,B,C)$ is real. The modular commutator is an intrinsically quantum quantity: $[K_{AB}, K_{BC}]$ is zero for classical states (\textit{i.e.}, states that are diagonal in the computational basis), and therefore so is $J$.

Thirdly, modular commutator is additive under tensor product. Namely, for state $\rho_{ABC}\otimes \lambda_{A'B'C'}$, we have
\begin{equation}\label{eq:additive}
	J(AA',BB',CC' )_{\rho\otimes \lambda } = J(A,B,C)_{\rho} + J(A',B',C')_{\lambda}.
\end{equation}

\subsection{Time reversal}

An important property of $J(A,B,C)$ is the fact that it is odd under bosonic time-reversal operation on the quantum state. Consider an antilinear map $*: \rho \to \rho^*$, where the entries of $\rho^*$ over a product basis over $A$, $B$, and $C$ are complex conjugates of the entries of $\rho$. One can show that
\begin{equation}
    \begin{aligned}
    J(A,B,C)_{\rho^{*}}&= i\text{Tr} (\rho_{ABC}^* [K_{AB}^{*}, K_{BC}^{*}] )\\
    &= i\left(\text{Tr}(\rho_{ABC} [K_{AB}, K_{BC}]) \right)^* \\
    &= i(-iJ(A,B,C)_{\rho})^* .
    \end{aligned}
\end{equation}
Recalling that $J$ is real, we obtain
\begin{equation}
\boxed{
J(A,B,C)_{\rho^{*}} = -J(A,B,C)_{\rho}.}
    \label{eq:odd_under_trs}
\end{equation}
A simple corollary is that $J(A,B,C)=0$ for any density matrix with real entries in some product basis over $A, B,$ and $C$.

An important implication of Eq.~\eqref{eq:odd_under_trs} is that $J(A,B,C)$ is zero for any reduced density matrix of a Gibbs state of a stoquastic Hamiltonian~\cite{Bravyi2006}. These are Hamiltonians which have nonnegative off-diagonal entries in the product basis over the local degrees of freedom, a famous example being the toric code Hamiltonian~\cite{Kitaev2003}. Therefore, in order for $J(A,B,C)$ to be nonzero, one must necessarily have a density matrix with complex-valued entries. However, the presence of complex entries is not sufficient. The modular commutator is invariant under the unitary conjugation
\begin{equation}
	\rho_{ABC} \to( U_A\otimes U_B\otimes U_C) \rho_{ABC} (U_A^{\dagger}\otimes U_B^{\dagger}\otimes U_C^{\dagger}), \nonumber
\end{equation}
where $U_A, U_B,$ and $U_C$ are unitary operators acting on $A, B,$ and $C$ respectively, and any states whose complex entries can be removed via such a conjugation must thus have zero $J(A,B,C)$. In this way, one can easily produce a density matrix with complex entries which nevertheless have a zero modular commutator.

\subsection{Quantum Markov chain}

There is a sense in which $J(A,B,C)$ detects a genuine tripartite correlation between $A,B,$ and $C$. This is because $J(A,B,C)=0$ if one of the subsystems is an empty set. This is trivially true if $B=\varnothing$ because $[K_{AB}, K_{BC}] = [K_A , K_C]=0$. If $A=\varnothing$, 
\begin{equation}
    \begin{aligned}
    J(\varnothing,B,C)_{\rho} &=i \text{Tr}(\rho_{BC}([K_B, K_{BC}]))\\
    &=i \text{Tr}(K_{BC}\rho_{BC}K_B) - i \text{Tr}(\rho_{BC} K_{BC} K_B) \\
    &= 0,
    \end{aligned}
\end{equation}
where in the second line we used the cyclicity of the trace and in the third line we used the fact that $\rho_{BC}$ and $K_{BC}$ commute with each other. Similarly, the same argument can be used to prove $J(A,B, \varnothing)=0$. To summarize, we have
\begin{equation}
    J(\varnothing, B, C)_{\rho}=     J(A, \varnothing, C)_{\rho}=     J(A, B, \varnothing)_{\rho}=0.
    \label{eq:zero_if_emptyset}
\end{equation}
Moreover, if $\rho_{ABC}$ forms a product state over any partition, $J(A,B,C)=0$. For instance, suppose $\rho_{ABC} = \rho_A \otimes \rho_{BC}$. Then $K_{AB} = K_A + K_B$, which leads to $J(A,B,C)=0$.  Other cases, \emph{e.g.,} $\rho_{ABC} = \rho_{AB}\otimes \rho_C$ and $\rho_{ABC} = \rho_B \otimes \rho_{AC}$, also lead to $J(A,B,C)=0$.

In fact, there is a more general set of conditions under which $J(A,B,C)$ is equal to $0$. These conditions can be concisely phrased in terms of the \emph{conditional mutual information}, defined as $I(A:C|B)_{\rho} := S(\rho_{AB}) + S(\rho_{BC}) - S(\rho_B)-S(\rho_{ABC})$, where $S(\rho):= -\text{Tr}(\rho \ln \rho)$ is the von Neumann entropy of $\rho$. It turns out that if $I(A:C|B)_{\rho}=0$ then $J(A,B,C)_{\rho}=0$. This fact can be proved easily by noting that~\cite{Petz2003}
\begin{equation}
    I(A:C|B)_{\rho}=0 \iff K_{AB} + K_{BC} - K_B - K_{ABC}=0 \label{eq:local_decomposition}
\end{equation}
for positive definite $\rho_{ABC}$.\footnote{If $\rho_{ABC}$ has zero eigenvalues, this equation means that it holds on the subspace spanned by the eigenstates of the nonzero eigenvalues of $\rho_{ABC}$.} Note that
\begin{equation}
\begin{aligned}
    J(A,B,C)_{\rho} &= \text{Tr}(\rho_{ABC}[K_{AB}, K_{BC}] ) \\
    &= \text{Tr}(\rho_{ABC}[K_{AB}, K_{ABC} +K_B - K_{AB}] ) \\
    &= \text{Tr}(\rho_{ABC}[K_{AB}, K_{ABC} +K_B] ) \\
    &=0
\end{aligned}
\end{equation}
because both $\Tr(\rho_{ABC}[K_{AB}, K_{ABC}])$ and $\Tr(\rho_{ABC}[K_{AB}, K_B])$ are zero for any density matrix $\rho_{ABC}$. Therefore, we conclude
\begin{equation}
\boxed{    I(A:C|B)_{\rho}=0 \implies J(A,B,C)_{\rho}=0.}
\label{eq:cmi_Jabc}
\end{equation}

Let us emphasize that Eq.~\eqref{eq:cmi_Jabc} holds for both bosons and fermions. This is because the condition $I(A:C|B)_{\rho}=0$ for fermions also implies $K_{ABC}=K_{AB} + K_{BC} - K_B - K_{ABC}$~\cite{Short}. 

\subsection{Modular flow}

One way to interpret $J(A,B,C)$ is to view it as a response of entanglement entropy under the modular flow. The modular Hamiltonian generates a modular flow, an automorphism on the algebra of operators acting on $A$:
\begin{equation}
O\to    O(s) = e^{iK_As}Oe^{-iK_As}. \label{eq:modular_flow}
\end{equation}
The modular flow played an important role in the development of Tomita-Takesaki theory~\cite{Haag2012}. The modular flow generated by $K_{BC}$, applied to a state, yields a one-parameter family of states  $\rho_{ABC}(s) := e^{-iK_{BC}s} \rho_{ABC} \, e^{iK_{BC}s}$. One can verify that
\begin{equation}
\begin{aligned}
\boxed{
    \frac{d}{ds}S(\rho_{AB}(s))\rvert_{s=0} =   J(A,B,C)_{\rho}. }
    \label{eq:entropy_under_flow}
\end{aligned}
\end{equation}
We leave the proof in Appendix~\ref{sec:appSJ}. 

\section{Area law}
\label{sec:area_law}

The discussion we had about $J(A,B,C)$ so far pertains to any quantum state. Now we shift our focus to a more special family of physical systems: gapped quantum many-body systems in two spatial dimensions. The ground state of a gapped system, which we denote as $\sigma$, is expected to obey the area law of entanglement entropy~\cite{Kitaev2006,Levin2006}. Area law means that the  following equation holds for any disk-shaped region $A$:
\begin{equation}
    S(\sigma_A) = \alpha |\partial A| - \gamma + \ldots, \label{eq:area_law}
\end{equation}
where $|\partial A|$ is the size of the perimeter of $A$, $\alpha$ is a non-universal constant, and $\gamma$ is the topological entanglement entropy~\cite{Kitaev2006,Levin2006}. The remaining term is expected to vanish in the $|A|\to \infty$ limit. 

For the set of regions depicted in Fig.~\ref{fig:abc}, the topological entanglement entropy~\cite{Kitaev2006,Levin2006} can be extracted as 
\begin{equation}\label{eq:teedef}
\gamma=\langle K_{AB}+K_{BC}+K_{CA}-K_A-K_B-K_C-K_{ABC}\rangle,  
\end{equation}
which consists of terms that are linear in the modular Hamiltonians. On the other hand, Eq.~\eqref{eq:main_result} is \emph{quadratic} in the modular Hamiltonian. There is, nevertheless, a connection between the entanglement entropy and the modular commutator, in the sense that Eq.~\eqref{eq:area_law} implies a nontrivial set of identities for the modular Hamiltonians, which subsequently constrains the properties of the modular commutator.

Recall that the vanishing conditional mutual information implies a nontrivial linear relation between the modular commutators (Eq.~\eqref{eq:local_decomposition}). Thus, by deriving $I(A:C|B)_{\sigma}=0$ for a judiciously chosen set of subsystems $A, B,$ and $C$ from Eq.~\eqref{eq:area_law}, we can derive nontrivial identities between modular commutators. We explain how this works in detail below.

Under Eq.~\eqref{eq:area_law}, the entanglement entropy obeys the following equation~\cite{SKK2019}:
\begin{equation} \label{eq:axioms}
    \begin{aligned}
    (S_{BC} + S_{CD} - S_B - S_D)_{\sigma}=0 
    \end{aligned}
\end{equation}
for regions which are topologically equivalent to the ones depicted in Fig.~\ref{fig:axioms}. Here $(S_X)_{\sigma}$ is a shorthand notation for $S(\sigma_X)$. (This entropy condition (\ref{eq:axioms}) is referred to as axiom \textbf{A1} in Ref.~\cite{SKK2019}.) This is the starting point of our derivation.
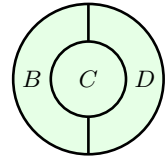
\begin{figure}[h]
    \centering
    \begin{tikzpicture}
	\begin{scope}[xshift=4.3cm]
	\draw[fill=green!10!white, line width=1pt] (0,0) circle (1);
	\draw[fill=green!10!white, line width=1pt] (0,0) circle (0.5);
	\node[] (B) at (-0.75,0) {$B$};
	\node[] (D) at (0.75,0) {$D$};
	\node[] (C) at (0,0) {$C$};
	\draw[line width=1pt] (0,0.5) -- (0,1);
	\draw[line width=1pt] (0,-0.5)--(0,-1);
	\end{scope}
    \end{tikzpicture}
    \caption{Axiom {\bf A1}: $(S_{BC} + S_{CD} - S_B - S_D)_{\sigma}=0$.}
    \label{fig:axioms}
\end{figure}

An important implication of Eq.~\eqref{eq:axioms} is that one can prove $I(A:C|B)_{\sigma}=0$ for \emph{any} $A$ which has no nontrivial overlap with $BCD$, \emph{i.e.}, $A\subset \Lambda \setminus (BCD)$. Let $E$ be the purifying space of $BCD$. Physically, $E$ includes the rest of the physical system that constitutes the two-dimensional quantum many-body system as well as the abstract ``environment'' that purifies $\sigma$. By using the purity of this global state, we conclude
\begin{equation}
    (S_{BC}+S_{CD} - S_B - S_D)_{\sigma} = I(E:C|B)_{|\psi\rangle_{BCDE}},
\end{equation}
where $|\psi\rangle_{BCDE}$ is the purified state. Clearly, $A$ is a subset of $E$. Therefore, 
\begin{equation}
    I(A:C|B)_{|\psi\rangle_{BCDE}} \leq I(E:C|D)_{|\psi\rangle_{BCDE}}.
\end{equation}
by using the strong subadditivity of entropy~\cite{Lieb1973}. Because the right hand side is $0$, again by the strong subadditivity of entropy, the left hand side must be zero as well. Since the reduced density matrix of $|\psi\rangle$ over $ABC$ is independent of the purification, we can conclude $I(A:C|B)_{\sigma}=0$. Let us summarize this finding as follows:
\begin{equation}
  (S_{BC} + S_{CD} - S_B - S_D)_{\sigma}=0 \implies I(A:C|B)_{\sigma}=0.
\label{eq:ssa_saturation}
\end{equation}
Furthermore, from Eq.~\eqref{eq:local_decomposition}, we can also conclude that 
\begin{equation}
    K_{AB} + K_{BC} - K_B - K_{ABC} = 0.
\end{equation}

\subsection{Topological invariance}
\label{sec:topological_invariance}

By using the modular Hamiltonain identities discussed above,  we can show that the modular commutator for the tripartition of a disk shown in Fig.~\ref{fig:abc} is a topological invariant. In other words,
\begin{equation}\label{eq:inv_copy}
		J(A,B,C)_{\sigma}\quad\text{on a partition of the form}\quad    \tikz[scale=0.2,baseline=0ex]{	
			\draw (0,0) circle (2.7);
			\draw (0,0) -- (30:2.7);
			\draw (0,0) -- (150:2.7);
			\draw (0,0) -- (-90:2.7);
			\node at (90:1.35) {\footnotesize{$B$}};
			\node at (-150:1.35) {\footnotesize{$C$}};
			\node at (-30:1.35) {\footnotesize{$A$}};
		}~
	\end{equation}
is invariant under smooth deformations of $A$, $B$, or $C$ as long as the topology is preserved.
(What follows is a review of Ref.~\cite{Short}, supplemented with more details.) To prove this statement, we shall consider small deformations of the subsystems $A$, $B$ and $C$ of the disk partition whilst keeping their topologies fixed. 

First we show that deformations away from $B$ leaves the modular commutator $J(A,B,C)_{\sigma}$ invariant.
Without loss of generality, we consider the two types of deformation of $A$, shown in Fig.~\ref{fig:topo_invariance_1}. From the area law Eq.~(\ref{eq:area_law}), we obtain
\begin{equation}\label{eq:inv_11}
	I(a:B\vert A)_{\sigma}=0.
\end{equation} 
The fact that this conditional mutual information is zero is equivalent to the fact that $K_{aAB} = K_{AB} + K_{aA} - K_{A}$; see Eq.~\eqref{eq:local_decomposition}, Thus we conclude 
\begin{equation}\label{eq:inv_12}
	\begin{aligned}
		&\,\,\,\,\,\,\,J(Aa,B,C)_{\sigma}\\
		&=i\langle[K_{aAB},K_{BC}]\rangle\\
		&=J(A,B,C)_{\sigma}+i\langle[K_{aA},K_{BC}]\rangle-i\langle[K_{A},K_{BC}]\rangle\,\\
		&=J(A,B,C)_{\sigma}.
	\end{aligned}	
\end{equation}
This argument also applies to the deformation of $C$ (instead of $A$).
\begin{figure}[h]
	\begin{tikzpicture}
		\begin{scope}[scale=0.3]
			\begin{scope}
				\fill[yellow, opacity=0.2] (-30:3.1) circle (0.75);
				\draw(-30:3.1) circle (0.75);
				
				\draw[fill=white] (0,0) circle (3);
				\node at (90:1.5) {$B$};
				\node at (-30:1.5) {$A$};
				\node at (-150:1.5) {$C$};
				\draw (0:0)--(30:3);
				\draw (0:0)--(150:3);
				\draw (0:0)--(-90:3);
				\node (a) at (-30:3.375) {$a$};
			\end{scope}	
			\begin{scope}[xshift=12 cm]
				\fill[yellow, opacity=0.2] (-90:3.1) circle (0.75);
				\draw(-90:3.1) circle (0.75);
				
				\draw[fill=white] (0,0) circle (3);
				\node at (90:1.5) {$B$};
				\node at (-30:1.5) {$A$};
				\node at (-150:1.5) {$C$};
				\draw (0:0)--(30:3);
				\draw (0:0)--(150:3);
				\draw (0:0)--(-90:3);
				\node (a) at (-90:3.378) {$a$};
			\end{scope}	
		\end{scope}	
	\end{tikzpicture}
	\caption{Subsystems involved in the proof of invariance against deformation of $A$ and $C$ while keeping the other subsystems ($BC$ and $AB$, respectively) intact. Here, $a \subset \Lambda\setminus (ABC)$. \label{fig:topo_invariance_1}}
\end{figure}
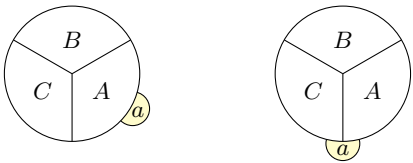

Second, we show that deformations of $B$ away from $A$ and $C$ leaves the modular commutator  $J(A,B,C)_{\sigma}$ invariant. This deformation is illustrated in Fig.~\ref{fig:topo_invariance_2}. For this case, the crucial facts are the following two quantum Markov chain conditions:
\begin{equation}\label{eq:inv_21}
	\begin{aligned}
		I(b:A\vert B)_{\sigma}=0,\\
		I(b:C\vert B)_{\sigma}=0.
	\end{aligned}
\end{equation}
Using an analysis similar to the one that leads to Eq.~(\ref{eq:inv_12}), we obtain
\begin{equation}\label{eq:inv_22}
	\begin{aligned}
		&\,\,\,\,\,\,\,J(A,Bb,C)_{\sigma}\\
		&=i\langle[K_{ABb},K_{BbC}]\rangle\\
		&=i\langle[K_{AB}+K_{Bb}-K_b,K_{BC}+K_{Bb}-K_b]\rangle\\
		&=J(A,B,C)_{\sigma}+i\langle[K_{AB},K_{Bb}-K_b]+[K_{Bb}-K_b,K_{BC}]\rangle\\
		&=J(A,B,C)_{\sigma},
	\end{aligned}	
\end{equation}
where we used $\langle [K_{AB},K_{Bb}] \rangle= \langle [K_{Bb},K_{BC}]\rangle =0$.
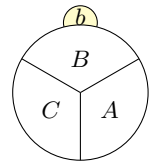
\begin{figure}[h]
	\begin{tikzpicture}
		\begin{scope}[scale=0.30]
			\begin{scope}
				\fill[yellow, opacity=0.2] (90:3.1) circle (0.75);
				\draw(90:3.1) circle (0.75);
				
				\draw[fill=white] (0,0) circle (3);
				\node at (90:1.5) {$B$};
				\node at (-30:1.5) {$A$};
				\node at (-150:1.5) {$C$};
				\draw (0:0)--(30:3);
				\draw (0:0)--(150:3);
				\draw (0:0)--(-90:3);
				\node (a) at (90:3.375) {$b$};
			\end{scope}	
		\end{scope}	
	\end{tikzpicture}
	\caption{Subsystems involved in the proof of invariance against deformation of $B$, while keeping $A$ and $C$ intact. Here $b \subset \Lambda\setminus (ABC)$ and $b$ is away from both $A$ and $C$. \label{fig:topo_invariance_2}}
\end{figure} 

Now there are two remaining cases, described in Fig.~\ref{fig:topo_invariance_3}. The proofs for these cases are less straightforward because the ``obvious'' choice of partitions do not form quantum Markov chains in general, \emph{e.g.,}  $I(b:A \vert B)_{\sigma}\ne 0$ and $I(b:C \vert B)_{\sigma}\ne 0$, for the left hand the right figure of Fig.~\ref{fig:topo_invariance_3}, respectively.

\begin{figure}[h]
	\begin{tikzpicture}
		\begin{scope}[scale=0.30]
			\begin{scope}
				\fill[yellow, opacity=0.2] (30:3.1) circle (0.75);
				\draw(30:3.1) circle (0.75);
				
				\draw[fill=white] (0,0) circle (3);
				\node at (90:1.5) {$B$};
				\node at (-30:1.5) {$A$};
				\node at (-150:1.5) {$C$};
				\draw (0:0)--(30:3);
				\draw (0:0)--(150:3);
				\draw (0:0)--(-90:3);
				\node (a) at (30:3.375) {$b$};
			\end{scope}	
			\begin{scope}[xshift=12 cm]
				\fill[yellow, opacity=0.2] (-30:0.1) circle (0.9);
				\draw (-30:0.1) circle (0.9);
				\fill[white] (0:0)--(30:2)--(150:2)--(-90:2)--(0:0)--cycle;
				
				\draw (0,0) circle (3);
				\node at (90:1.5) {$B$};
				\node at (-30:1.5) {$A$};
				\node at (-150:1.5) {$C$};
				\draw (0:0)--(30:3);
				\draw (0:0)--(150:3);
				\draw (0:0)--(-90:3);
				\node (a) at (-34:0.5) {$b$};
			\end{scope}	
		\end{scope}
	\end{tikzpicture}
	\caption{The last case for proving the topological invariance. (Left) $b \subset \Lambda\setminus (ABC)$ and $b$ is away from $C$. (Right) $b \subset A$, and $b$ is away from $D$ in the sense that it is contained in the interior of $ABC$.} \label{fig:topo_invariance_3}
\end{figure}
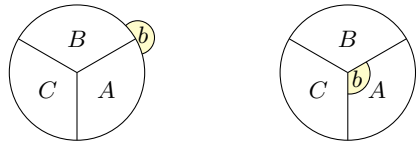

Below we solve these cases by using a quantum Markov chain involving a region $D$; here, $D$ is the complement of $ABC$ on system $\Lambda E$, where $E$ is a (finite dimensional) purifying system, such that $\sigma_{\Lambda}$ is purified to $\vert \psi_{ABCD}\rangle$. For each case shown in Fig.~\ref{fig:topo_invariance_3}, we have
\begin{equation}
	J(A,B,C)_{\vert \psi \rangle} = J(D,C,B)_{\vert \psi \rangle}
\end{equation}
which replaces the middle entry of modular commutator $B$ by $C$, so that the middle entry is now away from the place of deformation for the left figure. Similarly,
\begin{equation}
	J(A,B,C)_{\vert \psi \rangle} = J(C,D,A)_{\vert \psi \rangle},
\end{equation}
which replaces the middle entry $B$ by $D$, so that  the middle entry is now away from the place of deformation for the right figure. Now the problem can be solved in the same manner as the previous cases.
With the quantum Markov chain structure $I(b:C \vert D\setminus b)_{\vert \psi \rangle}=I(b:C|B)_{\vert \psi\rangle}=0$ for the left figure and $I(b:D \vert A \setminus b)_{\vert \psi \rangle}=0$ for the right figure, we prove the desired invariant property of $J(A,B,C)_{\sigma}$.

To conclude, given subsystems $A,B,$ and $C$ which are topologically equivalent to those in Fig.~\ref{fig:abc}, $J(A,B,C)_{\sigma}$ is invariant under any smooth deformation of $A,B,$ or $C$ that preserves the topology. (Note that $J(A,B,C)_{\sigma}$ is precisely the quantity that appears on the left hand side of Eq.~(\ref{eq:main_result}).)

\subsection{Non-orientable surface}

An immediate consequence of this topological invariance is that if the ground state is defined on a non-orientable surface ($\Lambda$ is non-orientable), then the modular commutator vanishes, \emph{i.e.,}
\begin{equation}\label{eq:inv_nonorientable}
	J(A,B,C)_{\sigma}=0\quad\text{on a face of the form}\quad    \tikz[scale=0.2,baseline=0ex]{	
		\draw (0,0) circle (2.7);
		\draw (0,0) -- (30:2.7);
		\draw (0,0) -- (150:2.7);
		\draw (0,0) -- (-90:2.7);
		\node at (90:1.35) {\footnotesize{$B$}};
		\node at (-150:1.35) {\footnotesize{$C$}};
		\node at (-30:1.35) {\footnotesize{$A$}};
	}~.
\end{equation}
(We note in passing that the following argument does not make use of the fact that $\sigma$ is a ground state. The only assumption used is the invariance of $J(A,B,C)_{\sigma}$ under continuous deformation of the subsystems.)  
This is because, on a non-orientable surface, one can smoothly deform the disk $ABC$ along a M\"{o}bius strip\footnote{To achieve this deformation, it is important to have axiom {\bf A1} holds on a set of bounded-radius disks which covers the M\"{o}bius strip.} such that, when it comes back, $A$ and $C$ are switched. The topological invariance then says $J(A,B,C)_{\sigma}=J(C,B,A)_{\sigma}$. On the other hand, $J(A,B,C)=-J(C,B,A)$ for any state. This proves Eq.~(\ref{eq:inv_nonorientable}).

\section{Locality of modular Hamiltonian}
\label{sec:locality_modular_operator}
Modular Hamiltonians are generally nonlocal, but they can be local in certain special cases. In any quantum field theory with Lorentz symmetry, $K_A$ is local if $A$ forms a half-space~\cite{BW1975,BW1976}. In conformal field theories, $K_A$ is local if $A$ is a disk~\cite{Casini2011}. The systems considered in this paper generally do not possess any such symmetries. As such, there seems no a priori reason to believe that the modular Hamiltonians in our setup would have a local structure. 

Nevertheless, it is possible to show that the modular Hamiltonians can be local (even without assuming any symmetry), just from Eq.~\eqref{eq:area_law}. For the derivation of Eq.~\eqref{eq:main_result}, the decomposition of the modular Hamiltonian on a disk will play an important role and this will be our primary focus. In fact, the argument is more general and can be applied to a disk containing an anyon (as opposed to a modular Hamiltonian obtained from the vacuum).

Whether a local decomposition exists or not depends greatly on the topology of the subsystem. For instance, If the subsystem is an annulus, we shall see that the modular Hamiltonian generally does not admit a local decomposition. 

For the ensuing discussion, it will be convenient to consider a sufficiently coarse-grained lattice, so that the correction term in Eq.~\eqref{eq:area_law} vanishes. Then we can apply an argument in Eq.~\eqref{eq:ssa_saturation} to conclude that certain conditional mutual information vanishes. (For example, $I(X:Z|Y)_{\sigma}=0$ would hold for the partition described in  Fig.~\ref{fig:lattice_partition}.) This implies --- via Eq.~\eqref{eq:local_decomposition} --- that the modular Hamiltonian can be broken down into a linear combination of modular Hamiltonians acting on smaller subsystems. By repeating this argument, we can decompose the modular Hamiltonian into a sum of operators each of which are localized.

\begin{figure}[h]
	\centering
	\begin{tikzpicture}
			\begin{scope}[scale=0.45]
		\begin{scope}
			\clip (2+0.25,2*0.866-0.433) rectangle (10-0.5*2+0.25, 10*0.866-0.433);
			\diskshadded{3}{2}{10}{2}{gray!60!white}{white};
			\diskshadded{4}{3}{10}{3}{gray!60!white}{white};
			\diskshadded{5}{4}{11}{4}{gray!60!white}{white};
			\diskshadded{5}{5}{11}{5}{gray!60!white}{white};
			\diskshadded{6}{6}{12}{6}{gray!60!white}{white};
			\diskshadded{6}{7}{12}{7}{gray!60!white}{white};
			\diskshadded{7}{8}{13}{8}{gray!60!white}{white};
			\diskshadded{7}{9}{14}{9}{gray!60!white}{white};
		\end{scope}
		\draw[] (2+0.25,2*0.866-0.433) rectangle (10-0.5*2+0.25, 10*0.866-0.433);
		\coloredboxhx{\zcolor,opacity=0.4}{9}{5}{9}{5};
		\coloredboxhx{\ycolor,opacity=0.4}{8}{5}{8}{5};
		\coloredboxhx{\xcolor,opacity=0.4}{8}{6}{8}{6};
	\end{scope}	
	\begin{scope}[xshift=4.8 cm, yshift=1.0 cm]
		\draw[fill=\xcolor, fill opacity=0.4] (3.5-3,1.3) rectangle (0.8, 1.6);
		\draw[fill=\ycolor, fill opacity=0.4] (3.5-3,0.9) rectangle (0.8, 1.2);
		\draw[fill=\zcolor, fill opacity=0.4] (3.5-3,0.5) rectangle (0.8, 0.8);
		\node [] (A) at (4.2-3,1.45) {$:X$};
		\node [] (B) at (4.2-3,1.05) {$:Y$};
		\node [] (C) at (4.2-3,0.65) {$:Z$};
	\end{scope}
\end{tikzpicture}
	\caption{Upon coarse-graining the local degrees of freedom, one obtains a triangular lattice of ``supersites''. (What is shown here is the dual of the triangular lattice.)  One can derive, using Eq.~\eqref{eq:ssa_saturation}, that  $I(X:Z|Y)_{\sigma}=0$.}
	\label{fig:lattice_partition}
\end{figure}
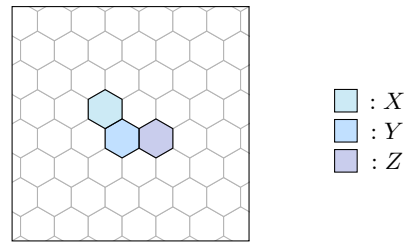

To make our discussion concrete, we set up our convention. Without loss of generality, partition the system into hexagonal cells, each containing a coarse-grained degree of freedom. We can represent each of these cells as a vertex of the dual lattice --- a triangular lattice. 

For the ensuing discussion, we shall define $F, E,$ and $V$ as the sets of faces, edges, and vertices of the triangular lattice.
We represent the elements of these sets as $f, e,$ and $v$. Let $A\subseteq V$ be a set of vertices. Define $A_{\partial} := \{v\in A: \exists (v,u)\in E \text{ s.t. }  u\not\in A\}$ and $A_{\text{int}} := A\setminus A_{\partial}$; see Fig.~\ref{fig:local_terms}(a).  Lastly, let $F(A)= \{(v,u,w)\in F: v,u,w\in A  \}$, and $E(A) = \{(u,v)\in E: u,v\in A \}$, and $N(v):= \{u\in V: (u,v) \in E \}$.

\subsection{Disk}
Let us first show that the modular Hamiltonian over a disk is local. The basic idea is to eliminate one site at a time from the disk, starting from its boundary. Specifically, for a disk $\mathfrak{D}$, consider a site $u\in \mathfrak{D}_{\partial}$. One can decompose $K_{\mathfrak{D}}$ as
\begin{equation}
    K_{\mathfrak{D}} = K_{\mathfrak{D} \setminus u} + K_{(N(u)\cap \mathfrak{D}) \cup \{ u\} } - K_{(N(u)\cap \mathfrak{D})}.
\end{equation}
In particular, we can decompose both $(N(u)\cap \mathfrak{D}) \cup \{u\}$ and  $(N(u)\cap \mathfrak{D})$ into a linear combination of modular Hamiltonians over faces, edges, and vertices. One can always choose the removed site such that the remaining subsystem is still a disk. Thus, one can repeat this argument, obtaining:
\begin{equation}\label{eq:mod_decomp}
	K_{\mathfrak{D}}=\sum_{\substack{f\in F(\mathfrak{D})}
	}K_{f}-\sum_{e\in E(\mathfrak{D})/E(\mathfrak{D}_{\partial})}K_{e}+\sum_{v\in\mathfrak{D}_{\text{int}}}K_{v}.
\end{equation}
(See Fig.~\ref{fig:local_terms}(b) for examples of $K_f$, $K_e$, and $K_v$.)

\begin{figure}[h]
	\centering
\begin{tikzpicture}
	\begin{scope}[scale=0.38]			
		\begin{scope}
			\clip (-2-0.25,-0.433) rectangle (10-0.5*8+0.25, 9*0.866-0.433);				
			\diskshadded{-2}{0}{6}{0}{gray!60!white}{white,opacity=0.4};
			\diskshadded{-1}{1}{0}{1}{gray!60!white}{white,opacity=0.4};
			\diskshadded{1}{1}{4}{1}{blue!50!cyan!30!white}{blue!50!cyan!10!white,opacity=0.4};
			\diskshadded{5}{1}{6}{1}{gray!60!white}{white,opacity=0.4};
			\diskshadded{-1}{2}{0}{2}{gray!60!white}{white,opacity=0.4};
			\diskshadded{1}{2}{5}{2}{blue!50!cyan!30!white}{blue!50!cyan!10!white,opacity=0.4};
			\diskshadded{6}{2}{7}{2}{gray!60!white}{white,opacity=0.4};
			\diskshadded{0}{3}{0}{3}{gray!60!white}{white,opacity=0.4};
			\diskshadded{1}{3}{6}{3}{blue!50!cyan!30!white}{blue!50!cyan!10!white,opacity=0.4};
			\diskshadded{7}{3}{7}{3}{gray!60!white}{white,opacity=0.4};
			\diskshadded{0}{4}{0}{4}{gray!60!white}{white,opacity=0.4};
			\diskshadded{1}{4}{7}{4}{blue!50!cyan!30!white}{blue!50!cyan!10!white,opacity=0.4};
			\diskshadded{8}{4}{8}{4}{gray!60!white}{white,opacity=0.4};
			\diskshadded{1}{5}{1}{5}{gray!60!white}{white,opacity=0.4};
			\diskshadded{2}{5}{7}{5}{blue!50!cyan!30!white}{blue!50!cyan!10!white,opacity=0.4};
			\diskshadded{8}{5}{8}{5}{gray!60!white}{white,opacity=0.4};
			\diskshadded{1}{6}{2}{6}{gray!60!white}{white,opacity=0.4};
			\diskshadded{3}{6}{7}{6}{blue!50!cyan!30!white}{blue!50!cyan!10!white,opacity=0.4};
			\diskshadded{8}{6}{9}{6}{gray!60!white}{white,opacity=0.4};
			\diskshadded{2}{7}{3}{7}{gray!60!white}{white,opacity=0.4};
			\diskshadded{4}{7}{7}{7}{blue!50!cyan!30!white}{blue!50!cyan!10!white,opacity=0.4};
			\diskshadded{8}{7}{9}{7}{gray!60!white}{white,opacity=0.4};
			\diskshadded{2}{8}{10}{8}{gray!60!white}{white,opacity=0.4};
			
			\diskshadded{1}{1}{4}{1}{blue!50!cyan!30!white}{blue!50!cyan!20!white,opacity=0.4};
			\diskshadded{1}{2}{1}{4}{blue!50!cyan!30!white}{blue!50!cyan!20!white,opacity=0.4};
			\diskshadded{7}{4}{7}{7}{blue!50!cyan!30!white}{blue!50!cyan!20!white,opacity=0.4};
			\diskshadded{4}{7}{6}{7}{blue!50!cyan!30!white}{blue!50!cyan!20!white,opacity=0.4};
			\diskshadded{3}{6}{3}{6}{blue!50!cyan!30!white}{blue!50!cyan!20!white,opacity=0.4};
			\diskshadded{2}{5}{2}{5}{blue!50!cyan!30!white}{blue!50!cyan!20!white,opacity=0.4};
			\diskshadded{5}{2}{5}{2}{blue!50!cyan!30!white}{blue!50!cyan!20!white,opacity=0.4};
			\diskshadded{6}{3}{6}{3}{blue!50!cyan!30!white}{blue!50!cyan!20!white,opacity=0.4};
			
			\draw[line width=0.8 pt] (30:0.578) -- (90:0.578) 
			-- ++(90:0.578)
			-- ++(150:0.578)-- ++(90:0.578) -- ++(150:0.578)-- ++(90:0.578) -- ++(150:0.578)-- ++(90:0.578) 
			-- ++(30:0.578)--++(90:0.578) -- ++(30:0.578)--++(90:0.578) -- ++(30:0.578)--++(90:0.578) -- ++(30:0.578)
			--++(-30:0.578) -- ++(30:0.578)--++(-30:0.578) -- ++(30:0.578)--++(-30:0.578) -- ++(30:0.578)--++(-30:0.578)
			-- ++(-90:0.578)--++(-30:0.578) -- ++(-90:0.578)--++(-30:0.578) -- ++(-90:0.578)--++(-30:0.578) -- ++(-90:0.578)
			--++(-150:0.578)--++(-90:0.578)--++(-150:0.578)--++(-90:0.578)--++(-150:0.578)--++(-90:0.578)--++(-150:0.578)
			--++(150:0.578)--++(210:0.578)--++(150:0.578)--++(210:0.578)--++(150:0.578)--++(210:0.578)--cycle;
			
			\draw[line width=0.8 pt, color=\ycolor!87!black] 
			++(90:0.578) ++(30:0.578)++(90:0.578)
			-- ++(90:0.578)
			-- ++(150:0.578)-- ++(90:0.578) -- ++(150:0.578)-- ++(90:0.578) 
			-- ++(30:0.578)--++(90:0.578) -- ++(30:0.578)--++(90:0.578) -- ++(30:0.578)
			--++(-30:0.578) -- ++(30:0.578)--++(-30:0.578) -- ++(30:0.578)--++(-30:0.578)
			-- ++(-90:0.578)--++(-30:0.578) -- ++(-90:0.578)--++(-30:0.578) -- ++(-90:0.578)
			--++(-150:0.578)--++(-90:0.578)--++(-150:0.578)--++(-90:0.578)--++(-150:0.578)
			--++(150:0.578)--++(210:0.578)--++(150:0.578)--++(210:0.578)--cycle;
			
			\node[]  at (2.2, 3.4){\large{$\mathfrak{D}_{\text{int}}$}};
			\node[]  at (5.45, 0.38){\large{$\mathfrak{D}_{\partial}$}};
			\draw[->, line width=0.7 pt, color=black] 
			(4.6, 0.7)
			.. controls +(180:0.3) and +(-80:0.4) ..
			(3.9, 1.9);
			\node[]  at (2, -1.3){{(a)}};
		\end{scope}
		\draw[] (-2-0.25,-0.433) rectangle (10-0.5*8+0.25, 9*0.866-0.433);
	\end{scope}
	
	\begin{scope}[xshift= 4 cm, scale=0.38]
		\begin{scope}
			\clip (-2-0.25,-0.433) rectangle (10-0.5*8+0.25, 9*0.866-0.433);
			
			\diskshadded{-2}{0}{6}{0}{gray!60!white}{white,opacity=0.4};
			\diskshadded{-1}{1}{0}{1}{gray!60!white}{white,opacity=0.4};
			\diskshadded{1}{1}{4}{1}{blue!50!cyan!30!white}{blue!50!cyan!10!white,opacity=0.4};
			\diskshadded{5}{1}{6}{1}{gray!60!white}{white,opacity=0.4};
			\diskshadded{-1}{2}{0}{2}{gray!60!white}{white,opacity=0.4};
			\diskshadded{1}{2}{5}{2}{blue!50!cyan!30!white}{blue!50!cyan!10!white,opacity=0.4};
			\diskshadded{6}{2}{7}{2}{gray!60!white}{white,opacity=0.4};
			\diskshadded{0}{3}{0}{3}{gray!60!white}{white,opacity=0.4};
			\diskshadded{1}{3}{6}{3}{blue!50!cyan!30!white}{blue!50!cyan!10!white,opacity=0.4};
			\diskshadded{7}{3}{7}{3}{gray!60!white}{white,opacity=0.4};
			\diskshadded{0}{4}{0}{4}{gray!60!white}{white,opacity=0.4};
			\diskshadded{1}{4}{7}{4}{blue!50!cyan!30!white}{blue!50!cyan!10!white,opacity=0.4};
			\diskshadded{8}{4}{8}{4}{gray!60!white}{white,opacity=0.4};
			\diskshadded{1}{5}{1}{5}{gray!60!white}{white,opacity=0.4};
			\diskshadded{2}{5}{7}{5}{blue!50!cyan!30!white}{blue!50!cyan!10!white,opacity=0.4};
			\diskshadded{8}{5}{8}{5}{gray!60!white}{white,opacity=0.4};
			\diskshadded{1}{6}{2}{6}{gray!60!white}{white,opacity=0.4};
			\diskshadded{3}{6}{7}{6}{blue!50!cyan!30!white}{blue!50!cyan!10!white,opacity=0.4};
			\diskshadded{8}{6}{9}{6}{gray!60!white}{white,opacity=0.4};
			\diskshadded{2}{7}{3}{7}{gray!60!white}{white,opacity=0.4};
			\diskshadded{4}{7}{7}{7}{blue!50!cyan!30!white}{blue!50!cyan!10!white,opacity=0.4};
			\diskshadded{8}{7}{9}{7}{gray!60!white}{white,opacity=0.4};
			\diskshadded{2}{8}{10}{8}{gray!60!white}{white,opacity=0.4};
			
			\draw[line width=0.8 pt] (30:0.578) -- (90:0.578) 
			-- ++(90:0.578)
			-- ++(150:0.578)-- ++(90:0.578) -- ++(150:0.578)-- ++(90:0.578) -- ++(150:0.578)-- ++(90:0.578) 
			-- ++(30:0.578)--++(90:0.578) -- ++(30:0.578)--++(90:0.578) -- ++(30:0.578)--++(90:0.578) -- ++(30:0.578)
			--++(-30:0.578) -- ++(30:0.578)--++(-30:0.578) -- ++(30:0.578)--++(-30:0.578) -- ++(30:0.578)--++(-30:0.578)
			-- ++(-90:0.578)--++(-30:0.578) -- ++(-90:0.578)--++(-30:0.578) -- ++(-90:0.578)--++(-30:0.578) -- ++(-90:0.578)
			--++(-150:0.578)--++(-90:0.578)--++(-150:0.578)--++(-90:0.578)--++(-150:0.578)--++(-90:0.578)--++(-150:0.578)
			--++(150:0.578)--++(210:0.578)--++(150:0.578)--++(210:0.578)--++(150:0.578)--++(210:0.578)--cycle;
			
			\coloredkv{yellow!60!white,opacity=0.5}{0}{6}{4};
			\coloredkf{green!60!yellow!50!white,opacity=0.5}{0}{1}{1};
			\coloredkf{green!60!yellow!50!white,opacity=0.5}{3}{2}{4};
			\coloredke{yellow!40!red!50!white,opacity=0.5}{1}{6}{6};
			\coloredke{yellow!40!red!50!white,opacity=0.5}{3}{4}{6};
			\coloredke{yellow!40!red!50!white,opacity=0.5}{2}{4}{3};
			
			\node[]  at (2, -1.3){{(b)}};
		\end{scope}
		\draw[] (-2-0.25,-0.433) rectangle (10-0.5*8+0.25, 9*0.866-0.433);
	\end{scope}
\end{tikzpicture}
	\caption{(a) A disk $\mathfrak{D}$ and its partition into $\mathfrak{D}_{\text{int}}$ and $\mathfrak{D}_{\partial}$. (b) A disk $\mathfrak{D}$ (blue) and the 1-site terms $K_v$ (yellow), 2-site terms $K_e$ (orange), 3 site terms $K_f$ (green) in the local decomposition (\ref{eq:mod_decomp}) of modular Hamiltonian $K_{\mathfrak{D}}$. 
	}
	\label{fig:local_terms}
\end{figure}
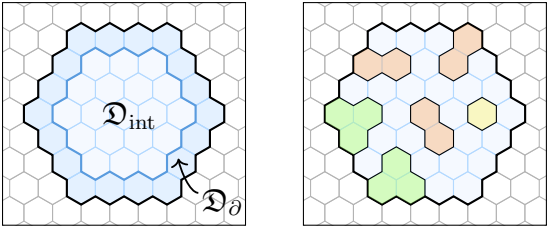

Alternatively, we can also derive a different decomposition which will be useful in Section~\ref{sec:modular_current}, \emph{i.e.,} $K_{\mathfrak{D}}= \sum_{v\in \mathfrak{D}}\widetilde{K}_{v}^{\mathfrak{D}}$, where
\begin{equation}\label{eq:mod_decomp2_addendum}
	\widetilde{K}_{v}^{\mathfrak{D}}=
	\begin{cases}
		{\displaystyle {\textstyle\frac{1}{3}}\sum_{f:v\in f}K_{f}-{\textstyle\frac{1}{2}}\sum_{e:v\in e}K_{e}+K_{v}}, & v\in\mathfrak{D}_{\text{int}}\\
		{\displaystyle {\textstyle\frac{1}{3}}\sum_{\substack{f:v\in f,\\
					f\in F(\mathfrak{D})
				}
			}K_{f}-{\textstyle\frac{1}{2}}\sum_{\substack{e:v\in e,\\
					e\in E(\mathfrak{D})\setminus E(\mathfrak{D}_{\partial})
				}
			}K_{e}}, & v\in\mathfrak{D}_{\partial}
	\end{cases}
	~.
\end{equation}
Note that $\mathfrak{D}$ in the superscript is not without meanings; each term, \emph{e.g.,} $\widetilde{K}_{v}^{\mathfrak{D}}$, depends on the choice of the disk $\mathfrak{D}$.

\subsection{Disk with an anyon}\label{ss:anyon_disk}

The decomposition in Eq.~\eqref{eq:mod_decomp} applies to the modular Hamiltonian of the \emph{ground state}. Here we show that a similar decomposition applies to excited states which contains a single anyon. For concreteness, consider a low-energy state with an anyon $a$ located in the interior of the disk $\mathfrak{D}$; see Fig.~\ref{fig:anyon_1}(a). We denote this state by $\sigma_{\mathfrak{D}}^{[a]}$. (Let us remark that we do not actually assume that $\sigma_{\mathfrak{D}}^{[a]}$ is the reduced density matrix of an eigenstate of some local Hamiltonian. The only fact we use is that the reduced density matrices of $\sigma_{\mathfrak{D}}^{[a]}$ and $\sigma_{\mathfrak{D}}$ are equal on any disk-shaped region $D\subset \mathfrak{D}$ that does not contain the anyon.)

The main complication in this context is that the axiom in Fig.~\ref{fig:axioms} does not always hold. If a disk contains a single anyon $a$, it has an extra contribution to the von Neumann entropy~\cite{Kitaev2006,SKK2019}:
\begin{equation}
 S(\sigma^{[a]}_{\mathfrak{D}})-S(\sigma_{\mathfrak{D}})= \ln d_a. \label{eq:entropy_quantum_dimension_correction}
\end{equation}
Therefore, if the anyon is located in region $C$ in Fig.~\ref{fig:axioms}, then $S_{BC}+S_{CD}-S_B-S_D= 2 \ln d_a$, which is generally nonzero. 
\begin{figure}[h]
	\centering
	\begin{tikzpicture}
		\begin{scope}[scale=1.2,xshift=-3 cm, yshift=1 cm]
			\draw[fill=blue!50!cyan!20!white, line width=0.8 pt] plot [smooth cycle] coordinates {(0:1.2) (30:1.2-0.1) (60:1.2-0.3) (90:1.2-0.15) (120:1.2)(150:1.2)(180:1.2-0.4)(210:1.2-0.8)(240:1.2-0.85)(270:1.2-0.6)(300:1.2-0.2)(330:1.2-0.1)};
			\draw[fill=red!50] (15:0.3) circle (0.07);
			\draw[]  (40:2) rectangle (220:1.7);
			\node[] (A) at (0.15,-1.35) {\footnotesize{(a)}};
			\node[]  at (0, 0.7){\large{$\mathfrak{D}$}};	
			\node[] at (0.5, 0) {$a$};	
		\end{scope}	
		\begin{scope}[scale=0.38]
			\begin{scope}
				\clip (-2-0.25,-0.433) rectangle (10-0.5*8+0.25, 9*0.866-0.433);
				\diskshadded{-2}{0}{6}{0}{gray!60!white}{white,opacity=0.4};
				\diskshadded{-1}{1}{0}{1}{gray!60!white}{white,opacity=0.4};
				\diskshadded{1}{1}{4}{1}{blue!50!cyan!30!white}{blue!50!cyan!10!white,opacity=0.4};
				\diskshadded{5}{1}{6}{1}{gray!60!white}{white,opacity=0.4};
				\diskshadded{-1}{2}{0}{2}{gray!60!white}{white,opacity=0.4};
				\diskshadded{1}{2}{5}{2}{blue!50!cyan!30!white}{blue!50!cyan!10!white,opacity=0.4};
				\diskshadded{6}{2}{7}{2}{gray!60!white}{white,opacity=0.4};
				\diskshadded{0}{3}{0}{3}{gray!60!white}{white,opacity=0.4};
				\diskshadded{1}{3}{6}{3}{blue!50!cyan!30!white}{blue!50!cyan!10!white,opacity=0.4};
				\diskshadded{7}{3}{7}{3}{gray!60!white}{white,opacity=0.4};
				\diskshadded{0}{4}{0}{4}{gray!60!white}{white,opacity=0.4};
				\diskshadded{1}{4}{7}{4}{blue!50!cyan!30!white}{blue!50!cyan!10!white,opacity=0.4};
				\diskshadded{8}{4}{8}{4}{gray!60!white}{white,opacity=0.4};
				\diskshadded{1}{5}{1}{5}{gray!60!white}{white,opacity=0.4};
				\diskshadded{2}{5}{7}{5}{blue!50!cyan!30!white}{blue!50!cyan!10!white,opacity=0.4};
				\diskshadded{8}{5}{8}{5}{gray!60!white}{white,opacity=0.4};
				\diskshadded{1}{6}{2}{6}{gray!60!white}{white,opacity=0.4};
				\diskshadded{3}{6}{7}{6}{blue!50!cyan!30!white}{blue!50!cyan!10!white,opacity=0.4};
				\diskshadded{8}{6}{9}{6}{gray!60!white}{white,opacity=0.4};
				\diskshadded{2}{7}{3}{7}{gray!60!white}{white,opacity=0.4};
				\diskshadded{4}{7}{7}{7}{blue!50!cyan!30!white}{blue!50!cyan!10!white,opacity=0.4};
				\diskshadded{8}{7}{9}{7}{gray!60!white}{white,opacity=0.4};
				\diskshadded{2}{8}{10}{8}{gray!60!white}{white,opacity=0.4};
			\end{scope}
			\draw[] (-2-0.25,-0.433) rectangle (10-0.5*8+0.25, 9*0.866-0.433);
			\draw[line width=0.8 pt] (30:0.578) -- (90:0.578) 
			-- ++(90:0.578)
			-- ++(150:0.578)-- ++(90:0.578) -- ++(150:0.578)-- ++(90:0.578) -- ++(150:0.578)-- ++(90:0.578) 
			-- ++(30:0.578)--++(90:0.578) -- ++(30:0.578)--++(90:0.578) -- ++(30:0.578)--++(90:0.578) -- ++(30:0.578)
			--++(-30:0.578) -- ++(30:0.578)--++(-30:0.578) -- ++(30:0.578)--++(-30:0.578) -- ++(30:0.578)--++(-30:0.578)
			-- ++(-90:0.578)--++(-30:0.578) -- ++(-90:0.578)--++(-30:0.578) -- ++(-90:0.578)--++(-30:0.578) -- ++(-90:0.578)
			--++(-150:0.578)--++(-90:0.578)--++(-150:0.578)--++(-90:0.578)--++(-150:0.578)--++(-90:0.578)--++(-150:0.578)
			--++(150:0.578)--++(210:0.578)--++(150:0.578)--++(210:0.578)--++(150:0.578)--++(210:0.578)--cycle;
			\draw[fill=red!50] (30:0.578*4)++(90:0.578*4) circle (0.21);
			\node[] at (2.75,0.578*6) {$a$};
			\node[]  at (2.0, 5.4){\large{$\mathfrak{D}$}};
			\node[]  at (2.0, -1.15){\footnotesize{(b)}};
		\end{scope}	
	\end{tikzpicture}
	\caption{(a) A disk $\mathfrak{D}$ containing an anyon. (b) Without loss of generality, we can coarse-grain the system in such a way that the anyon is located on a single cell. (In the dual lattice, this anyon lives on the vertex.)}\label{fig:anyon_1}
\end{figure}
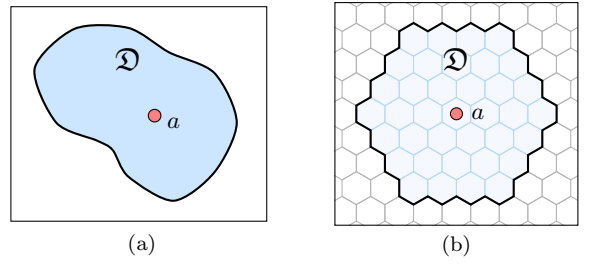

However, we can circumvent this issue by computing the conditional mutual information directly using Eq.~\eqref{eq:entropy_quantum_dimension_correction}. For concreteness, without loss of generality, assume that the anyon is located on a vertex of the coarse-grained lattice; see Fig.~\ref{fig:anyon_1}(b) for an illustration.   
Suppose $XY, YZ, Y, XYZ$ all have disk topology, and an anyon $a$ is located in $Y$. Then the net contribution of entropy in $S_{XY}+S_{YZ}-S_Y-S_{XYZ}$, due to the  anyon $a$, is zero, ($\ln d_a+ \ln d_a -\ln d_a - \ln d_a=0$). For example, the choices of $X,Y,Z$ below all have $I(X:Z \vert Y)_{\sigma^{[a]}}=0$, where the red dot is anyon $a$:
  \begin{equation}
  \begin{tikzpicture}
  	\begin{scope}[scale=0.45]
  		\begin{scope}
  			\clip (3-0.5*2-0.25,2*0.866-0.433) rectangle (16-0.5*6+0.25, 7*0.866-0.433);
  			\diskshadded{3}{2}{14}{2}{gray!60!white}{white,opacity=0.4};
  			\diskshadded{4}{3}{14}{3}{gray!60!white}{white,opacity=0.4};
  			\diskshadded{4}{4}{15}{4}{gray!60!white}{white,opacity=0.4};
  			\diskshadded{5}{5}{15}{5}{gray!60!white}{white,opacity=0.4};
  			\diskshadded{5}{6}{16}{6}{gray!60!white}{white,opacity=0.4};
  			\begin{scope}
  				\coloredboxhx{\zcolor,opacity=0.4}{5}{3}{5}{3};
  				\coloredboxhx{\ycolor,opacity=0.4}{6}{4}{6}{4};
  				\coloredboxhx{\xcolor,opacity=0.4}{7}{5}{7}{5};
  				
  				\draw[fill=red!50] (30:0.578*8)++(90:0.578*2) circle (0.21);
  				
  				\coloredboxhx{\zcolor,opacity=0.4}{9}{3}{9}{3};
  				\coloredboxhx{\zcolor,opacity=0.4}{8}{3}{8}{3};

  				\coloredboxhx{\ycolor,opacity=0.4}{9}{4}{9}{4};
  				\coloredboxhx{\ycolor,opacity=0.4}{10}{4}{10}{4};
  				
  				\coloredboxhx{\xcolor,opacity=0.4}{10}{5}{10}{5};
  				
  				\draw[fill=red!50] (30:0.578*14)++(90:-0.578) circle (0.21);
  				
  				\coloredboxhx{\zcolor,opacity=0.4}{14}{4}{14}{4};
  				\coloredboxhx{\zcolor,opacity=0.4}{13}{3}{13}{3};
  				
  				\coloredboxhx{\ycolor,opacity=0.4}{12}{3}{12}{3};
  				\coloredboxhx{\ycolor,opacity=0.4}{13}{4}{13}{4};
  				\coloredboxhx{\ycolor,opacity=0.4}{14}{5}{14}{5};
  				
  				\coloredboxhx{\xcolor,opacity=0.4}{13}{5}{13}{5};
  				\coloredboxhx{\xcolor,opacity=0.4}{12}{4}{12}{4};
  				
  				\draw[fill=red!50] (30:0.578*22)++(90:-0.578*5) circle (0.21);
  			\end{scope}
  			\draw[] (3-0.5*2-0.25,2*0.866-0.433) rectangle (16-0.5*6+0.25, 7*0.866-0.433);
  		\end{scope}
  		
  	\end{scope}	
  	\begin{scope}[xshift=6.5 cm, yshift=0.5 cm]
  		\draw[fill=\xcolor, fill opacity=0.4] (3.5-3,1.3) rectangle (0.8, 1.6);
  		\draw[fill=\ycolor, fill opacity=0.4] (3.5-3,0.9) rectangle (0.8, 1.2);
  		\draw[fill=\zcolor, fill opacity=0.4] (3.5-3,0.5) rectangle (0.8, 0.8);
  		\node [] (A) at (4.2-3,1.45) {$:X$};
  		\node [] (B) at (4.2-3,1.05) {$:Y$};
  		\node [] (C) at (4.2-3,0.65) {$:Z$};
  	\end{scope}
  \end{tikzpicture}
  \end{equation}
 Similarly, the net contribution of the anyon to the conditional mutual information is zero if the anyon is located in $X$ or $Z$, leading to the same quantum Markov chain condition $I(X:Z \vert Y)_{\sigma^{[a]}}=0$. 
 Thus, the same local decomposition (\ref{eq:mod_decomp}) holds for the modular Hamiltonian of $\sigma^{[a]}_{\mathfrak{D}}$.

\subsection{Annulus}

Unlike the modular Hamiltonian of a disk, the modular Hamiltonian of an annulus may not admit a local decomposition, even if the area law (Eq.~\eqref{eq:area_law}) holds. However, there is another state on the annulus which is locally indistinguishable from the ground state such that its logarithm admits a local decomposition. (In the special case in which the topological entanglement entropy $\gamma$ is $0$, this state \emph{is} the reduced density matrix of the ground state on the annulus.)

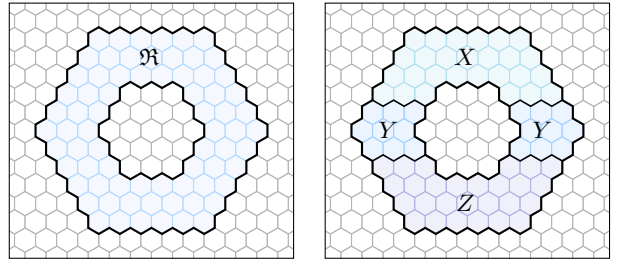
\begin{figure}[h]
	\centering
	\begin{tikzpicture}
	\begin{scope}[scale=0.28]
		\begin{scope}
			\clip(-3-0.5-0.25,-0.866) rectangle (13-0.5*6-0.25, 13*0.866);
			\diskshadded{-4}{-1}{10}{-1}{gray!60!white}{white,opacity=0.4};
			\diskshadded{-4}{0}{10}{0}{gray!60!white}{white,opacity=0.4};
			\diskshadded{-4}{1}{0}{1}{gray!60!white}{white,opacity=0.4};
			\diskshadded{7}{1}{11}{1}{gray!60!white}{white,opacity=0.4};
			\diskshadded{1}{1}{6}{1}{blue!50!cyan!30!white}{blue!50!cyan!10!white,opacity=0.4};
			\diskshadded{-3}{2}{0}{2}{gray!60!white}{white,opacity=0.4};
			\diskshadded{8}{2}{11}{2}{gray!60!white}{white,opacity=0.4};
			\diskshadded{1}{2}{7}{2}{blue!50!cyan!30!white}{blue!50!cyan!10!white,opacity=0.4};
			\diskshadded{-3}{3}{0}{3}{gray!60!white}{white,opacity=0.4};
			\diskshadded{9}{3}{12}{3}{gray!60!white}{white,opacity=0.4};
			\diskshadded{1}{3}{8}{3}{blue!50!cyan!30!white}{blue!50!cyan!10!white,opacity=0.4};
			\diskshadded{-2}{4}{0}{4}{gray!60!white}{white,opacity=0.4};
			\diskshadded{4}{4}{6}{4}{gray!60!white}{white,opacity=0.4};
			\diskshadded{10}{4}{12}{4}{gray!60!white}{white,opacity=0.4};
			\diskshadded{1}{4}{3}{4}{blue!50!cyan!30!white}{blue!50!cyan!10!white,opacity=0.4};
			\diskshadded{7}{4}{9}{4}{blue!50!cyan!30!white}{blue!50!cyan!10!white,opacity=0.4};
			\diskshadded{-2}{5}{0}{5}{gray!60!white}{white,opacity=0.4};
			\diskshadded{4}{5}{7}{5}{gray!60!white}{white,opacity=0.4};
			\diskshadded{11}{5}{13}{5}{gray!60!white}{white,opacity=0.4};
			\diskshadded{1}{5}{3}{5}{blue!50!cyan!30!white}{blue!50!cyan!10!white,opacity=0.4};
			\diskshadded{8}{5}{10}{5}{blue!50!cyan!30!white}{blue!50!cyan!10!white,opacity=0.4};
			\diskshadded{-1}{6}{0}{6}{gray!60!white}{white,opacity=0.4};
			\diskshadded{4}{6}{8}{6}{gray!60!white}{white,opacity=0.4};
			\diskshadded{12}{6}{13}{6}{gray!60!white}{white,opacity=0.4};
			\diskshadded{1}{6}{3}{6}{blue!50!cyan!30!white}{blue!50!cyan!10!white,opacity=0.4};
			\diskshadded{9}{6}{11}{6}{blue!50!cyan!30!white}{blue!50!cyan!10!white,opacity=0.4};
			\diskshadded{-1}{7}{1}{7}{gray!60!white}{white,opacity=0.4};
			\diskshadded{5}{7}{8}{7}{gray!60!white}{white,opacity=0.4};
			\diskshadded{12}{7}{14}{7}{gray!60!white}{white,opacity=0.4};
			\diskshadded{2}{7}{4}{7}{blue!50!cyan!30!white}{blue!50!cyan!10!white,opacity=0.4};
			\diskshadded{9}{7}{11}{7}{blue!50!cyan!30!white}{blue!50!cyan!10!white,opacity=0.4};
			\diskshadded{0}{8}{2}{8}{gray!60!white}{white,opacity=0.4};
			\diskshadded{5}{8}{8}{8}{gray!60!white}{white,opacity=0.4};
			\diskshadded{12}{8}{14}{8}{gray!60!white}{white,opacity=0.4};
			\diskshadded{3}{8}{5}{8}{blue!50!cyan!30!white}{blue!50!cyan!10!white,opacity=0.4};
			\diskshadded{9}{8}{11}{8}{blue!50!cyan!30!white}{blue!50!cyan!10!white,opacity=0.4};
			
			\diskshadded{0}{9}{3}{9}{gray!60!white}{white,opacity=0.4};
			\diskshadded{12}{9}{15}{9}{gray!60!white}{white,opacity=0.4};
			\diskshadded{4}{9}{11}{9}{blue!50!cyan!30!white}{blue!50!cyan!10!white,opacity=0.4};
			\diskshadded{1}{10}{4}{10}{gray!60!white}{white,opacity=0.4};
			\diskshadded{12}{10}{15}{10}{gray!60!white}{white,opacity=0.4};
			\diskshadded{5}{10}{11}{10}{blue!50!cyan!30!white}{blue!50!cyan!10!white,opacity=0.4};
			\diskshadded{1}{11}{5}{11}{gray!60!white}{white,opacity=0.4};
			\diskshadded{12}{11}{16}{11}{gray!60!white}{white,opacity=0.4};
			\diskshadded{6}{11}{11}{11}{blue!50!cyan!30!white}{blue!50!cyan!10!white,opacity=0.4};
			\diskshadded{2}{12}{16}{13}{gray!60!white}{white,opacity=0.4};
		\end{scope}
		\draw (-3-0.5-0.25,-0.866) rectangle (13-0.5*6-0.25, 13*0.866);

		\draw[line width=0.8 pt] (30:0.578) -- (90:0.578) 
		-- ++(90:0.578) -- ++(150:0.578)-- ++(90:0.578) -- ++(150:0.578)-- ++(90:0.578) -- ++(150:0.578)-- ++(90:0.578) -- ++(150:0.578)-- ++(90:0.578) -- ++(150:0.578)
		--++(90:0.578) -- ++(30:0.578)--++(90:0.578) -- ++(30:0.578)--++(90:0.578) -- ++(30:0.578)--++(90:0.578) -- ++(30:0.578)--++(90:0.578) -- ++(30:0.578)--++(90:0.578) -- ++(30:0.578)
		--++(-30:0.578) -- ++(30:0.578)--++(-30:0.578) -- ++(30:0.578)--++(-30:0.578) -- ++(30:0.578)--++(-30:0.578) -- ++(30:0.578)--++(-30:0.578) -- ++(30:0.578)--++(-30:0.578) 
		-- ++(-90:0.578)--++(-30:0.578) -- ++(-90:0.578)--++(-30:0.578) -- ++(-90:0.578)--++(-30:0.578) -- ++(-90:0.578)--++(-30:0.578) -- ++(-90:0.578)--++(-30:0.578) 
		--++(-90:0.578)--++(-150:0.578)--++(-90:0.578)--++(-150:0.578)--++(-90:0.578)--++(-150:0.578)--++(-90:0.578)--++(-150:0.578)--++(-90:0.578)--++(-150:0.578)--++(-90:0.578)--++(-150:0.578)
		--++(150:0.578)--++(210:0.578)--++(150:0.578)--++(210:0.578)--++(150:0.578)--++(210:0.578)--++(150:0.578)--++(210:0.578)--++(150:0.578)--++(210:0.578)--cycle;	
		\draw[line width=0.8 pt] 
		++(90:0.578*3) ++(30:0.578*3)++(90:0.578)
		-- ++(90:0.578)
		-- ++(150:0.578)-- ++(90:0.578) -- ++(150:0.578)-- ++(90:0.578) 
		-- ++(30:0.578)--++(90:0.578) -- ++(30:0.578)--++(90:0.578) -- ++(30:0.578)
		--++(-30:0.578) -- ++(30:0.578)--++(-30:0.578) -- ++(30:0.578)--++(-30:0.578)
		-- ++(-90:0.578)--++(-30:0.578) -- ++(-90:0.578)--++(-30:0.578) -- ++(-90:0.578)
		--++(-150:0.578)--++(-90:0.578)--++(-150:0.578)--++(-90:0.578)--++(-150:0.578)
		--++(150:0.578)--++(210:0.578)--++(150:0.578)--++(210:0.578)--cycle;

		\node[]  at (2.9, 8.7){{{$\mathfrak{R}$}}};		
	\end{scope}	
	
	\begin{scope}[xshift=4.2 cm, scale=0.28]
		\begin{scope}
			\clip(-3-0.5-0.25,-0.866) rectangle (13-0.5*6-0.25, 13*0.866);
			\diskshadded{-4}{-1}{10}{-1}{gray!60!white}{white,opacity=0.4};
			\diskshadded{-4}{0}{10}{0}{gray!60!white}{white,opacity=0.4};
			\diskshadded{-4}{1}{0}{1}{gray!60!white}{white,opacity=0.4};
			\diskshadded{7}{1}{11}{1}{gray!60!white}{white,opacity=0.4};
			\diskshadded{1}{1}{6}{1}{\zcolor!60!white}{\zcolor!30!white,opacity=0.4};
			\diskshadded{-3}{2}{0}{2}{gray!60!white}{white,opacity=0.4};
			\diskshadded{8}{2}{11}{2}{gray!60!white}{white,opacity=0.4};
			\diskshadded{1}{2}{7}{2}{\zcolor!60!white}{\zcolor!30!white,opacity=0.4};
			\diskshadded{-3}{3}{0}{3}{gray!60!white}{white,opacity=0.4};
			\diskshadded{9}{3}{12}{3}{gray!60!white}{white,opacity=0.4};
			\diskshadded{1}{3}{8}{3}{\zcolor!60!white}{\zcolor!30!white,opacity=0.4};
			\diskshadded{-2}{4}{0}{4}{gray!60!white}{white,opacity=0.4};
			\diskshadded{4}{4}{6}{4}{gray!60!white}{white,opacity=0.4};
			\diskshadded{10}{4}{12}{4}{gray!60!white}{white,opacity=0.4};
			\diskshadded{1}{4}{3}{4}{\zcolor!60!white}{\zcolor!30!white,opacity=0.4};
			\diskshadded{7}{4}{9}{4}{\zcolor!60!white}{\zcolor!30!white,opacity=0.4};
			\diskshadded{-2}{5}{0}{5}{gray!60!white}{white,opacity=0.4};
			\diskshadded{4}{5}{7}{5}{gray!60!white}{white,opacity=0.4};
			\diskshadded{11}{5}{13}{5}{gray!60!white}{white,opacity=0.4};
			\diskshadded{1}{5}{3}{5}{\ycolor!60!white}{\ycolor!30!white,opacity=0.4};
			\diskshadded{8}{5}{10}{5}{\ycolor!60!white}{\ycolor!30!white,opacity=0.4};
			\diskshadded{-1}{6}{0}{6}{gray!60!white}{white,opacity=0.4};
			\diskshadded{4}{6}{8}{6}{gray!60!white}{white,opacity=0.4};
			\diskshadded{12}{6}{13}{6}{gray!60!white}{white,opacity=0.4};
			\diskshadded{1}{6}{3}{6}{\ycolor!60!white}{\ycolor!30!white,opacity=0.4};
			\diskshadded{9}{6}{11}{6}{\ycolor!60!white}{\ycolor!30!white,opacity=0.4};
			\diskshadded{-1}{7}{1}{7}{gray!60!white}{white,opacity=0.4};
			\diskshadded{5}{7}{8}{7}{gray!60!white}{white,opacity=0.4};
			\diskshadded{12}{7}{14}{7}{gray!60!white}{white,opacity=0.4};
			\diskshadded{2}{7}{4}{7}{\ycolor!60!white}{\ycolor!30!white,opacity=0.4};
			\diskshadded{9}{7}{11}{7}{\ycolor!60!white}{\ycolor!30!white,opacity=0.4};
			\diskshadded{0}{8}{2}{8}{gray!60!white}{white,opacity=0.4};
			\diskshadded{5}{8}{8}{8}{gray!60!white}{white,opacity=0.4};
			\diskshadded{12}{8}{14}{8}{gray!60!white}{white,opacity=0.4};
			\diskshadded{3}{8}{5}{8}{\xcolor!60!white}{\xcolor!30!white,opacity=0.4};
			\diskshadded{9}{8}{11}{8}{\xcolor!60!white}{\xcolor!30!white,opacity=0.4};
			
			\diskshadded{0}{9}{3}{9}{gray!60!white}{white,opacity=0.4};
			\diskshadded{12}{9}{15}{9}{gray!60!white}{white,opacity=0.4};
			\diskshadded{4}{9}{11}{9}{\xcolor!60!white}{\xcolor!30!white,opacity=0.4};
			\diskshadded{1}{10}{4}{10}{gray!60!white}{white,opacity=0.4};
			\diskshadded{12}{10}{15}{10}{gray!60!white}{white,opacity=0.4};
			\diskshadded{5}{10}{11}{10}{\xcolor!60!white}{\xcolor!30!white,opacity=0.4};
			\diskshadded{1}{11}{5}{11}{gray!60!white}{white,opacity=0.4};
			\diskshadded{12}{11}{16}{11}{gray!60!white}{white,opacity=0.4};
			\diskshadded{6}{11}{11}{11}{\xcolor!60!white}{\xcolor!30!white,opacity=0.4};
			\diskshadded{2}{12}{16}{13}{gray!60!white}{white,opacity=0.4};
		\end{scope}
		\draw (-3-0.5-0.25,-0.866) rectangle (13-0.5*6-0.25, 13*0.866);
		
		\draw[line width=0.6 pt] 
		++(90:0.578*4) ++(150:0.578*3)++(90:0.578)
		--++(30:0.578)
		-- ++(-30:0.578)-- ++(30:0.578) -- ++(-30:0.578)-- ++(30:0.578); 
		\draw[line width=0.6 pt] 
		++(90:0.578*10) ++(150:0.578*3)
		-- ++(-30:0.578)
		-- ++(30:0.578)-- ++(-30:0.578) -- ++(30:0.578)-- ++(-30:0.578); 
		
		\draw[line width=0.6 pt] 
		++(90:0.578*4) ++(150:0.578*3)++(90:0.578)++(30:0.578*7)++(-30:0.578*6)
		--++(-30:0.578)
		-- ++(30:0.578)-- ++(-30:0.578) -- ++(30:0.578)-- ++(-30:0.578); 
		\draw[line width=0.6 pt] 
		++(90:0.578*10) ++(150:0.578*3)++(30:0.578*6)++(-30:0.578*7)
		--++(30:0.578)
		-- ++(-30:0.578)-- ++(30:0.578) -- ++(-30:0.578)-- ++(30:0.578); 
		\draw[line width=0.8 pt] (30:0.578) -- (90:0.578) 
		-- ++(90:0.578) -- ++(150:0.578)-- ++(90:0.578) -- ++(150:0.578)-- ++(90:0.578) -- ++(150:0.578)-- ++(90:0.578) -- ++(150:0.578)-- ++(90:0.578) -- ++(150:0.578)
		--++(90:0.578) -- ++(30:0.578)--++(90:0.578) -- ++(30:0.578)--++(90:0.578) -- ++(30:0.578)--++(90:0.578) -- ++(30:0.578)--++(90:0.578) -- ++(30:0.578)--++(90:0.578) -- ++(30:0.578)
		--++(-30:0.578) -- ++(30:0.578)--++(-30:0.578) -- ++(30:0.578)--++(-30:0.578) -- ++(30:0.578)--++(-30:0.578) -- ++(30:0.578)--++(-30:0.578) -- ++(30:0.578)--++(-30:0.578) 
		-- ++(-90:0.578)--++(-30:0.578) -- ++(-90:0.578)--++(-30:0.578) -- ++(-90:0.578)--++(-30:0.578) -- ++(-90:0.578)--++(-30:0.578) -- ++(-90:0.578)--++(-30:0.578) 
		--++(-90:0.578)--++(-150:0.578)--++(-90:0.578)--++(-150:0.578)--++(-90:0.578)--++(-150:0.578)--++(-90:0.578)--++(-150:0.578)--++(-90:0.578)--++(-150:0.578)--++(-90:0.578)--++(-150:0.578)
		--++(150:0.578)--++(210:0.578)--++(150:0.578)--++(210:0.578)--++(150:0.578)--++(210:0.578)--++(150:0.578)--++(210:0.578)--++(150:0.578)--++(210:0.578)--cycle;	
		\draw[line width=0.8 pt] 
		++(90:0.578*3) ++(30:0.578*3)++(90:0.578)
		-- ++(90:0.578)
		-- ++(150:0.578)-- ++(90:0.578) -- ++(150:0.578)-- ++(90:0.578) 
		-- ++(30:0.578)--++(90:0.578) -- ++(30:0.578)--++(90:0.578) -- ++(30:0.578)
		--++(-30:0.578) -- ++(30:0.578)--++(-30:0.578) -- ++(30:0.578)--++(-30:0.578)
		-- ++(-90:0.578)--++(-30:0.578) -- ++(-90:0.578)--++(-30:0.578) -- ++(-90:0.578)
		--++(-150:0.578)--++(-90:0.578)--++(-150:0.578)--++(-90:0.578)--++(-150:0.578)
		--++(150:0.578)--++(210:0.578)--++(150:0.578)--++(210:0.578)--cycle;

		\node[]  at (2.9, 8.7){{$X$}};	
		\node[]  at (2.9, 8.7-12*0.578){{$Z$}};	
		\node[]  at (2.9-3.66, 8.7-6*0.578){{$Y$}};	
		\node[]  at (2.9+3.66, 8.7-6*0.578){{$Y$}};		
	\end{scope}	
	\end{tikzpicture}
\caption{An annulus and its Levin-Wen partition $\mathfrak{R}=XYZ$.}\label{fig:annulus}
\end{figure}

Specifically, for an annulus $\mathfrak{R}$ (see Fig.~\ref{fig:annulus}), consider a Hamiltonian in the spirit of Eq.~\eqref{eq:mod_decomp}:  
\begin{equation}\label{eq:mod_decomp_annuli_1}
	K_{\mathfrak{R}}':=\sum_{\substack{f\in F(\mathfrak{R})}
	}K_{f}-\sum_{e\in E(\mathfrak{R})/E(\mathfrak{R}_{\partial})}K_{e}+\sum_{v\in\mathfrak{R}_{\text{int}}}K_{v}~,
\end{equation}
where each term in Eq.~(\ref{eq:mod_decomp_annuli_1}) is the modular Hamiltonian obtained from the ground state $\sigma$. The left hand side is, in general, \emph{not} equal to the modular Hamiltonian of $\mathfrak{R}$, which we denote as $K_{\mathfrak{R}}(\sigma)$. (For clarity, here we will make the dependence of the modular Hamiltonian to $\sigma$ explicit by denoting it as $K_{\mathfrak{R}}(\sigma)$ instead of $K_{\mathfrak{R}}$.) The ``Gibbs state'' of the Hamiltonian $K_{\mathfrak{R}}'$, \emph{i.e.,}
\begin{equation}
	\rho^{\text{max}}_{\mathfrak{R}}:=e^{-K_{\mathfrak{R}}'} \label{eq:annulus_mod_decomp}
\end{equation}
is actually the maximum-entropy state on $\mathfrak{R}$ among all the states that are locally indistinguishable from the ground state $\sigma$~\cite{SKK2019}.

Let us consider the Levin-Wen partition~\cite{Levin2006} of the annulus $\mathfrak{R}=XYZ$, shown in Fig.~\ref{fig:annulus}. One can verify with Eq.~\eqref{eq:mod_decomp}, that
\begin{equation}
	K_{\mathfrak{R}}'=K_{XY}(\sigma) + K_{YZ}(\sigma)- K_Y(\sigma).
\end{equation}
This is because $XY$ and $YZ$ are both disks, and that $Y$ is a union of two disks on which $\sigma_Y$ factorizes due to Eq.~\eqref{eq:area_law}. A simple calculation shows that 
\begin{equation}\label{eq:maxandgs}
	\begin{aligned}
			S(\sigma_{\mathfrak{R}}\|\rho^{\text{max}}_{\mathfrak{R}})&= I(X:Z|Y)_\sigma\,\\
			&= 2\gamma,
	\end{aligned}
\end{equation}
where $S(\rho\|\sigma):=\Tr\rho(\ln \rho-\ln \sigma)$ is the relative entropy. The $\gamma$ in the second line is the topological entanglement entropy. 

When the state $\sigma$ has zero topological entanglement entropy, Eq.~\eqref{eq:maxandgs} implies $S(\sigma_{\mathfrak{R}}\|\rho^{\text{max}}_{\mathfrak{R}})=0$, and thus the modular Hamiltonian is local: $K_{\mathfrak{R}}(\sigma)=K'_{\mathfrak{R}}$. 
However, if $\sigma$ has nonzero topological entanglement entropy, $S(\sigma_{\mathfrak{R}}\|\rho^{\text{max}}_{\mathfrak{R}})>0$, \emph{i.e.}, $\sigma_{\mathfrak{R}}$ and $\rho^{\text{max}}_{\mathfrak{R}}$ are different states. Moreover, we see that $K_{\mathfrak{R}}(\sigma)$ cannot have \emph{any} local decomposition in this case. This is because ${\rho}^{\text{max}}_{\mathfrak{R}}$ is the closest state to $\sigma_{\mathfrak{R}}$ in the set of the Gibbs states defined by all local Hamiltonians~\cite{Kato2015,Kato2019} (see Theorem 3 of Ref.~\cite{Kato2019} in particular).

\section{Modular current}
\label{sec:modular_current}

Because the modular Hamiltonian for a disk $\mathfrak{D}$ is local, one can formally view it as a local Hamiltonian. Then, the reduced density matrix of the ground state, $\sigma_{\mathfrak{D}}=e^{-K_{\mathfrak{D}}}$, can be viewed as a thermal state of this local Hamiltonian at temperature $T=1$. In this Section, we shall study the ``energy current'' of this Hamiltonian on this ``thermal state". 

Let us review general facts about energy current in the many-body context. Let $H=\sum_u h_u$ be a many-body Hamiltonian where the index $u$ represents a site in a set $\Lambda$ and $h_u$ is a self-adjoint operator that acts nontrivially on a ball of bounded radius centered at $u$. The energy current of this Hamiltonian at temperature $T$, from $u$ to $v$, is defined as $i \text{Tr}(e^{-H/T}Z^{-1} [h_u, h_v])$, where $Z=\text{Tr}(e^{-H/T})$ is the partition function~\cite{Kitaev2006a}. More generally, the energy current from $A$ to $B$, (where $A,B\subset \Lambda$ and $A\cap B= \emptyset$) is defined as $\sum_{u\in A,v \in B}\text{Tr}(e^{-H/T}Z^{-1} [h_u, h_v])$. There are two immediate facts that follow from this definition. First, the energy current between two sufficiently well-separated sites vanishes. This is simply because $h_u$ and $h_v$ will be acting nontrivially on two disjoint sets if $u$ and $v$ are sufficiently far apart, leading to $[h_u, h_v]=0$. Second, the energy current is conserved, \emph{i.e.,}
\begin{equation}
\begin{aligned}
    \sum_{u} i \text{Tr}(e^{-H/T}Z^{-1} [h_u, h_v])  &= \text{Tr}(e^{-H/T}Z^{-1} [H, h_v]) \\
    &=0,
\end{aligned}
\label{eq:mod_current_conservation}
\end{equation}
where we used the cyclicity of the trace in the last step.

The energy current associated with the modular Hamiltonian can be defined in an analogous way. We shall refer to this ``energy current" as the \emph{modular current}. Concretely, using  the decomposition $K_{\mathfrak{D}}= \sum_{v\in \mathfrak{D}} \widetilde{K}_v^{\mathfrak{D}}$ (see Eq.~\eqref{eq:mod_decomp2_addendum}), we define the modular current from $v$ to $u$ as: 
\begin{equation}
	f^{\mathfrak{D}}_{vu} := i\langle [\widetilde{K}^{\mathfrak{D}}_v, \widetilde{K}^{\mathfrak{D}}_u]\rangle.
\end{equation} 
See Fig.~\ref{fig:triple_point}(a) for an illustration. Because the modular Hamiltonian is formally a local Hamiltonian, the modular current is also local and conserved.

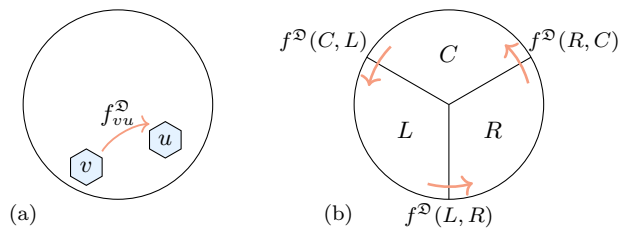
\begin{figure}[h]
	\centering 
	\begin{tikzpicture}
		\begin{scope}[scale=0.42]
			\begin{scope}
				\fill[white] (0:0) circle (3);
				\draw (0,0) circle (3);
				\begin{scope}[yshift=-0.25 cm]
					\diskshadded{1}{-1}{1}{-1}{black}{blue!50!cyan!25!white,opacity=0.4};
					\diskshadded{-2}{-2}{-2}{-2}{black}{blue!50!cyan!25!white,opacity=0.4};
					\node at (2-0.5,-0.866) {$u$};
					\node at (-1,-0.866*2) {$v$};
					\node at (0:0) {{${f^{\mathfrak{D}}_{vu}}$}};
					\draw[->,semithick,line width=0.25mm, color=red!50!orange!93!black!50!white] (-1+0.5,-0.866*2+0.5) .. controls +(60:0.5) and +(180:0.5) .. (2-0.5-0.5,-0.866+0.5);
				\end{scope}
				\node[] (A) at  (-3.0,-3.5) {\footnotesize{(a)}};
			\end{scope}	
		\end{scope}		
		\begin{scope}[xshift=4.4 cm, scale=0.42]
			\begin{scope}
				\fill[white] (0:0) circle (3);
				\draw (0,0) circle (3);
				\draw (0:0)--(30:3);
				\draw (0:0)--(150:3);
				\draw (0:0)--(-90:3);
				
				\node at (90:1.6) {$C$};
				\node at (-30:1.6) {$R$};
				\node at (-150:1.6) {$L$};
				
				\node at  (27:4.45)  {\footnotesize{$f^{\mathfrak{D}}(R,C)$}};
				\node at  (153:4.45) {\footnotesize{$f^{\mathfrak{D}}(C,L)$}};
				\node at  (270:3.5) {\footnotesize{$f^{\mathfrak{D}}(L,R)$}};
				
				\node[] (A) at (-3.5,-3.5) {\footnotesize{(b)}};
				\draw[->,semithick,line width=0.35mm, color=red!50!orange!93!black!50!white] (0+15:2.6) arc (0+15:60-13:2.6);
				\draw[->,semithick,line width=0.35mm, color=red!50!orange!93!black!50!white] (120+15:2.6) arc (120+15:180-13:2.6);
				\draw[->,semithick,line width=0.35mm, color=red!50!orange!93!black!50!white] (240+15:2.6) arc (240+15:300-13:2.6);
			\end{scope}	
		\end{scope}			
	\end{tikzpicture}
	\caption{(a) Two sites within a disk, $v,u\subset \mathfrak{D}$, and the modular current from $v$ to $u$. (b) The disk $\mathfrak{D}$ is divided into $L$, $R$, and the remainder $C=\mathfrak{D}\setminus (LR)$. $f^{\mathfrak{D}}(L,R) =f^{\mathfrak{D}}(R,C) =f^{\mathfrak{D}}(C,L)$ follows from the conservation.}
	\label{fig:triple_point}
\end{figure}

More generally, we shall define the modular current from subsystem $L$ to subsystem $R$ (such that $L,R\subset \mathfrak{D}$ and $L \cap R=\emptyset$) as
\begin{equation}
	f^{\mathfrak{D}}(L,R) :=\sum_{v\in L} \sum_{u \in R} f^{\mathfrak{D}}_{vu}.
\end{equation}
It follows directly from the definition that the modular current satisfies the following useful properties:
\begin{equation}\label{eq:fLR_properties}
	\begin{aligned}
		f(L,R)&=-f(R,L),\\
		f^{\mathfrak{D}}(L_1 L_2,R)&=  f^{\mathfrak{D}}(L_1,R)+f^{\mathfrak{D}}(L_2,R),\\
		f^{\mathfrak{D}}(L,R_1R_2)&= f^{\mathfrak{D}}(L,R_1)+f^{\mathfrak{D}}(L,R_2),\\
		f^{\mathfrak{D}}(L, \mathfrak{D}\setminus L)&=0,\\
		f^{\mathfrak{D}}(L,R) =f^{\mathfrak{D}}(R,C) &=f^{\mathfrak{D}}(C,L), \textrm{ for any }LRC=\mathfrak{D}.
	\end{aligned}
\end{equation}
See Fig.~\ref{fig:triple_point}(b) for one choice of $L,R,C$ for the last case.

What we discussed so far are general properties of the modular current (so long as the modular Hamiltonian admits a local decomposition). However, more can be said about the modular current by utilizing the discussion in Section~\ref{sec:modular_commutator}, \ref{sec:area_law} and \ref{sec:locality_modular_operator}. Remarkably, any modular current between two subsystems can be determined by a modular current between a \emph{single} pair of sites multiplied by a rational number determined by the geometry of the subsystems; see Section~\ref{subsection:simplify}. Furthermore, the modular current vanishes in the bulk, leading to a conserved edge current flowing along the edge; see Section~\ref{subsection:edge_modular_current}. 

\subsection{Simplifying the modular current}\label{subsection:simplify}

Below, we will provide a series of arguments that make the calculation of the modular current progressively simpler. Each term $\widetilde{K}_{u}^{\mathfrak{D}}$ is a linear combination of modular Hamiltonians acting on faces, edges, and vertices; see Eq.~(\ref{eq:mod_decomp2_addendum}). For the purpose calculating the modular current, the vertex terms are irrelevant. Consider a vertex term associated with a vertex $v$ ($K_v$ of Fig.~\ref{fig:local_terms}) for concreteness. The key point is that $v$ is either a strict subset of $A\subset \Lambda$ or disjoint from $A$. Either way, modular commutator $i\langle [K_v, K_A]\rangle$ vanishes because of Eq.~\eqref{eq:zero_if_emptyset}. 

What about the terms associated with the edges or faces ($K_e$, $K_f$ of Fig.~\ref{fig:local_terms})? Commutators involving some terms of this type are zero. The most trivial case is when their respective supports are disjoint. Another case is when their supports overlap but in such a way that the non-overlapping parts are not adjacent to each other.
For instance, for the two faces $f$ and $f'$ shown in Fig.~\ref{fig:vanishing_mod_comm}, we have quantum Markov chain condition  $I(f\setminus f' : f'\setminus f \vert f \cap f')_{\sigma}=0$, which follows from the area law. Then, by Eq.~\eqref{eq:cmi_Jabc}, we conclude that the modular commutator $i\langle [K_f,K_{f'}]\rangle=J(f\setminus f',f \cap f',f'\setminus f)_{\sigma}=0$.

\begin{figure}[h]
	\centering
\begin{tikzpicture}
	\begin{scope}[scale=0.38]
		\begin{scope}
			\clip (-2-0.25,-0.433) rectangle (10-0.5*8+0.25, 9*0.866-0.433);
			
			\diskshadded{-2}{0}{6}{0}{gray!60!white}{white,opacity=0.4};
			\diskshadded{-1}{1}{0}{1}{gray!60!white}{white,opacity=0.4};
			\diskshadded{1}{1}{4}{1}{blue!50!cyan!30!white}{blue!50!cyan!10!white,opacity=0.4};
			\diskshadded{5}{1}{6}{1}{gray!60!white}{white,opacity=0.4};
			\diskshadded{-1}{2}{0}{2}{gray!60!white}{white,opacity=0.4};
			\diskshadded{1}{2}{5}{2}{blue!50!cyan!30!white}{blue!50!cyan!10!white,opacity=0.4};
			\diskshadded{6}{2}{7}{2}{gray!60!white}{white,opacity=0.4};
			\diskshadded{0}{3}{0}{3}{gray!60!white}{white,opacity=0.4};
			\diskshadded{1}{3}{6}{3}{blue!50!cyan!30!white}{blue!50!cyan!10!white,opacity=0.4};
			\diskshadded{7}{3}{7}{3}{gray!60!white}{white,opacity=0.4};
			\diskshadded{0}{4}{0}{4}{gray!60!white}{white,opacity=0.4};
			\diskshadded{1}{4}{7}{4}{blue!50!cyan!30!white}{blue!50!cyan!10!white,opacity=0.4};
			\diskshadded{8}{4}{8}{4}{gray!60!white}{white,opacity=0.4};
			\diskshadded{1}{5}{1}{5}{gray!60!white}{white,opacity=0.4};
			\diskshadded{2}{5}{7}{5}{blue!50!cyan!30!white}{blue!50!cyan!10!white,opacity=0.4};
			\diskshadded{8}{5}{8}{5}{gray!60!white}{white,opacity=0.4};
			\diskshadded{1}{6}{2}{6}{gray!60!white}{white,opacity=0.4};
			\diskshadded{3}{6}{7}{6}{blue!50!cyan!30!white}{blue!50!cyan!10!white,opacity=0.4};
			\diskshadded{8}{6}{9}{6}{gray!60!white}{white,opacity=0.4};
			\diskshadded{2}{7}{3}{7}{gray!60!white}{white,opacity=0.4};
			\diskshadded{4}{7}{7}{7}{blue!50!cyan!30!white}{blue!50!cyan!10!white,opacity=0.4};
			\diskshadded{8}{7}{9}{7}{gray!60!white}{white,opacity=0.4};
			\diskshadded{2}{8}{10}{8}{gray!60!white}{white,opacity=0.4};
			
			\draw[line width=0.8 pt] (30:0.578) -- (90:0.578) 
			-- ++(90:0.578)
			-- ++(150:0.578)-- ++(90:0.578) -- ++(150:0.578)-- ++(90:0.578) -- ++(150:0.578)-- ++(90:0.578) 
			-- ++(30:0.578)--++(90:0.578) -- ++(30:0.578)--++(90:0.578) -- ++(30:0.578)--++(90:0.578) -- ++(30:0.578)
			--++(-30:0.578) -- ++(30:0.578)--++(-30:0.578) -- ++(30:0.578)--++(-30:0.578) -- ++(30:0.578)--++(-30:0.578)
			-- ++(-90:0.578)--++(-30:0.578) -- ++(-90:0.578)--++(-30:0.578) -- ++(-90:0.578)--++(-30:0.578) -- ++(-90:0.578)
			--++(-150:0.578)--++(-90:0.578)--++(-150:0.578)--++(-90:0.578)--++(-150:0.578)--++(-90:0.578)--++(-150:0.578)
			--++(150:0.578)--++(210:0.578)--++(150:0.578)--++(210:0.578)--++(150:0.578)--++(210:0.578)--cycle;
			
			\coloredkf{green!60!yellow!50!white,opacity=0.5}{0}{2}{3};
			
			\node[]  at (2, -1.3){{(b)}};
		\end{scope}
		\draw[] (-2-0.25,-0.433) rectangle (10-0.5*8+0.25, 9*0.866-0.433);
		\node[]  at (1, 4*0.866-0.433){$f$};
			
	\end{scope}
	
	\begin{scope}[xshift= 4 cm, scale=0.38]
		\begin{scope}
			\clip (-2-0.25,-0.433) rectangle (10-0.5*8+0.25, 9*0.866-0.433);
			
			\diskshadded{-2}{0}{6}{0}{gray!60!white}{white,opacity=0.4};
			\diskshadded{-1}{1}{0}{1}{gray!60!white}{white,opacity=0.4};
			\diskshadded{1}{1}{4}{1}{blue!50!cyan!30!white}{blue!50!cyan!10!white,opacity=0.4};
			\diskshadded{5}{1}{6}{1}{gray!60!white}{white,opacity=0.4};
			\diskshadded{-1}{2}{0}{2}{gray!60!white}{white,opacity=0.4};
			\diskshadded{1}{2}{5}{2}{blue!50!cyan!30!white}{blue!50!cyan!10!white,opacity=0.4};
			\diskshadded{6}{2}{7}{2}{gray!60!white}{white,opacity=0.4};
			\diskshadded{0}{3}{0}{3}{gray!60!white}{white,opacity=0.4};
			\diskshadded{1}{3}{6}{3}{blue!50!cyan!30!white}{blue!50!cyan!10!white,opacity=0.4};
			\diskshadded{7}{3}{7}{3}{gray!60!white}{white,opacity=0.4};
			\diskshadded{0}{4}{0}{4}{gray!60!white}{white,opacity=0.4};
			\diskshadded{1}{4}{7}{4}{blue!50!cyan!30!white}{blue!50!cyan!10!white,opacity=0.4};
			\diskshadded{8}{4}{8}{4}{gray!60!white}{white,opacity=0.4};
			\diskshadded{1}{5}{1}{5}{gray!60!white}{white,opacity=0.4};
			\diskshadded{2}{5}{7}{5}{blue!50!cyan!30!white}{blue!50!cyan!10!white,opacity=0.4};
			\diskshadded{8}{5}{8}{5}{gray!60!white}{white,opacity=0.4};
			\diskshadded{1}{6}{2}{6}{gray!60!white}{white,opacity=0.4};
			\diskshadded{3}{6}{7}{6}{blue!50!cyan!30!white}{blue!50!cyan!10!white,opacity=0.4};
			\diskshadded{8}{6}{9}{6}{gray!60!white}{white,opacity=0.4};
			\diskshadded{2}{7}{3}{7}{gray!60!white}{white,opacity=0.4};
			\diskshadded{4}{7}{7}{7}{blue!50!cyan!30!white}{blue!50!cyan!10!white,opacity=0.4};
			\diskshadded{8}{7}{9}{7}{gray!60!white}{white,opacity=0.4};
			\diskshadded{2}{8}{10}{8}{gray!60!white}{white,opacity=0.4};
			
			\draw[line width=0.8 pt] (30:0.578) -- (90:0.578) 
			-- ++(90:0.578)
			-- ++(150:0.578)-- ++(90:0.578) -- ++(150:0.578)-- ++(90:0.578) -- ++(150:0.578)-- ++(90:0.578) 
			-- ++(30:0.578)--++(90:0.578) -- ++(30:0.578)--++(90:0.578) -- ++(30:0.578)--++(90:0.578) -- ++(30:0.578)
			--++(-30:0.578) -- ++(30:0.578)--++(-30:0.578) -- ++(30:0.578)--++(-30:0.578) -- ++(30:0.578)--++(-30:0.578)
			-- ++(-90:0.578)--++(-30:0.578) -- ++(-90:0.578)--++(-30:0.578) -- ++(-90:0.578)--++(-30:0.578) -- ++(-90:0.578)
			--++(-150:0.578)--++(-90:0.578)--++(-150:0.578)--++(-90:0.578)--++(-150:0.578)--++(-90:0.578)--++(-150:0.578)
			--++(150:0.578)--++(210:0.578)--++(150:0.578)--++(210:0.578)--++(150:0.578)--++(210:0.578)--cycle;
			
			\coloredkf{green!60!yellow!50!white,opacity=0.5}{1}{1}{2};
			
			\node[]  at (2, -1.3){{(b)}};
		\end{scope}
		\draw[] (-2-0.25,-0.433) rectangle (10-0.5*8+0.25, 9*0.866-0.433);
		\node[]  at (0, 3*0.866-0.433){$f'$};
	\end{scope}
\end{tikzpicture}
	\caption{Two green faces $f, f' \subset \mathfrak{D}$, for which the corresponding modular commutator $i\langle [K_f,K_{f'}]\rangle$ vanishes.
	}
	\label{fig:vanishing_mod_comm}
\end{figure}
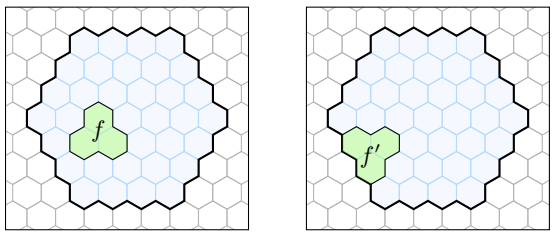

Thus, the modular commutator is nonzero only if the ground state of the subsystems associated with the respective modular Hamiltonians do not define a quantum Markov chain, in a way discussed above. Furthermore, by the topological invariance of the modular commutator, all these nonzero terms are equal up to a sign. Thus, the commutators of any of the terms ($K_f,K_e,K_v$) attain a value of either $0$ or $\pm J$, where 
\begin{equation}\label{eq:alpha_def}
	J:=J(u,v,w)_{\sigma} \quad\text{on a face of the form}\quad    \tikz[scale=0.4,baseline=5ex]{	
		\coloredboxhx{white,opacity=0.5}{6}{2}{7}{2};
		\coloredboxhx{white,opacity=0.5}{7}{3}{7}{3};
		\node at (5,1.73) {$w$};
		\node at (5.5,2.55) {$v$};
		\node at (6,1.73) {$u$};
	}~.
\end{equation}
Let us emphasize that this expression is independent of the face we are choosing. This is because the modular commutator is a topological invariant, as discussed in Section~\ref{sec:topological_invariance}.

\subsection{Edge modular current}\label{subsection:edge_modular_current}

In this section, we observe that the modular current flows along the edge and vanishes in the bulk. This leads to a well-defined edge modular current; we further calculate its precise value: $I_{\sigma}= J/4$ in Eq.~(\ref{eq:edge_1/4J}). After that, we arrive at our formula of chiral central charge by a physical argument. 

First, as we discussed,  the modular current $f^{\mathfrak{D}}_{vu}$ is both conserved and local. On the coarse-grained geometry, the locality is explicit : every individual $\widetilde{K}_{u}^{\mathfrak{D}}$ acts nontrivially on a ball of radius $1$. Therefore, the modular current from $u$ to $v$ can be nonzero only if their distance is less or equal to $2$. We shall only consider those cases.

Let us begin with the modular current between a pair of sites in the bulk, \emph{i.e.,} $\mathfrak{D}_{\text{int}}$. Using two facts we have already established, it is easy to show that the modular current vanishes between any two such sites: (i) modular commutator $J(A,B,C)$ is antisymmetric under the exchange of $A$ and $C$ and (ii) $J(A,B,C)$ is invariant under smooth deformation. The key point is that one can express the modular current as a sum of modular commutators wherein the $A$ and $C$ appear twice in the sum, with their order exchanged. Naturally, the sum vanishes and so does the modular current. In summary,
\begin{equation}\label{eq:vanishes_in_bulk}
	f_{vu}^{\mathfrak{D}} =0, \quad \forall v,u \subset \mathfrak{D}_{\text{int}}.
\end{equation}

Together with the locality and conservation properties of the modular current, Eq.~(\ref{eq:vanishes_in_bulk}) yields a conserved \emph{edge modular current} flows continuously along the edge.  Moreover, deforming part of the edge cannot change the current on another portion of the edge; this can be seen explicitly from the
local decomposition \eqref{eq:mod_decomp2_addendum}. Therefore, the edge modular current must be invariant under all deformations of disk $\mathfrak{D}$, since such deformations can be built up from local deformations. For this reason, we shall use $I_{\sigma}$ to denote the edge modular current, dropping the index $\mathfrak{D}$. 

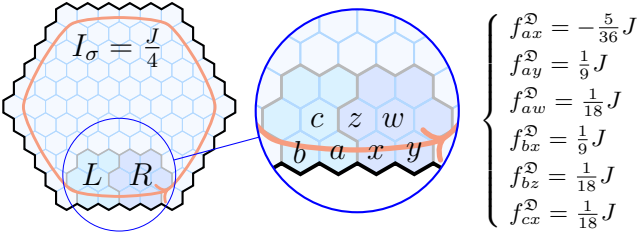
\begin{figure}[h]
	\centering
	\begin{tikzpicture}
		\begin{scope}[scale=0.8]
			\begin{scope}[scale=0.35,spy using outlines={circle, magnification=1.8, size=1 cm, connect spies}]
				
				\diskshadded{1}{1}{6}{1}{blue!50!cyan!30!white}{blue!50!cyan!10!white,opacity=0.4};
				\diskshadded{2}{1}{3}{1}{blue!50!cyan!30!white}{blue!30!cyan!30!white,opacity=0.4};
				\diskshadded{4}{1}{5}{1}{blue!50!cyan!30!white}{blue!70!cyan!30!white,opacity=0.4};
				\diskshadded{1}{2}{7}{2}{blue!50!cyan!30!white}{blue!50!cyan!10!white,opacity=0.4};
				\diskshadded{2}{2}{3}{2}{blue!50!cyan!30!white}{blue!30!cyan!30!white,opacity=0.4};
				\diskshadded{4}{2}{6}{2}{blue!50!cyan!30!white}{blue!70!cyan!30!white,opacity=0.4};
				\diskshadded{1}{3}{8}{3}{blue!50!cyan!30!white}{blue!50!cyan!10!white,opacity=0.4};
				\diskshadded{3}{3}{4}{3}{blue!50!cyan!30!white}{blue!30!cyan!30!white,opacity=0.4};
				\diskshadded{5}{3}{6}{3}{blue!50!cyan!30!white}{blue!70!cyan!30!white,opacity=0.4};
				\diskshadded{1}{4}{9}{4}{blue!50!cyan!30!white}{blue!50!cyan!10!white,opacity=0.4};
				\diskshadded{1}{5}{10}{5}{blue!50!cyan!30!white}{blue!50!cyan!10!white,opacity=0.4};
				\diskshadded{1}{6}{11}{6}{blue!50!cyan!30!white}{blue!50!cyan!10!white,opacity=0.4};
				
				\diskshadded{2}{7}{11}{7}{blue!50!cyan!30!white}{blue!50!cyan!10!white,opacity=0.4};
				\diskshadded{3}{8}{11}{8}{blue!50!cyan!30!white}{blue!50!cyan!10!white,opacity=0.4};
				\diskshadded{4}{9}{11}{9}{blue!50!cyan!30!white}{blue!50!cyan!10!white,opacity=0.4};
				\diskshadded{5}{10}{11}{10}{blue!50!cyan!30!white}{blue!50!cyan!10!white,opacity=0.4};
				\diskshadded{6}{11}{11}{11}{blue!50!cyan!30!white}{blue!50!cyan!10!white,opacity=0.4};

				\draw[line width=0.6 pt, color=gray!55!white] (30:0.578) -- ++(30:0.578) 
				-- ++(90:0.578) -- ++(150:0.578)-- ++(90:0.578) -- ++(30:0.578)-- ++(90:0.578) -- ++(30:0.578) -- ++(-30:0.578)-- ++(30:0.578)
				-- ++(-30:0.578)-- ++(30:0.578)-- ++(-30:0.578)-- ++(30:0.578)
				--++(-30:0.578)--++(-90:0.578)--++(-30:0.578)--++(-90:0.578)--++(-150:0.578)--++(-90:0.578);

				\draw[line width=0.6 pt, color=gray!55!white] (30:0.578) ++ (0:2)++(30:0.578)--++ (90:0.578)--++(150:0.578)--++(90:0.578)--++(30:0.578)--++(90:0.578);
				
				\draw[line width=0.8 pt] (30:0.578) -- (90:0.578) 
				-- ++(90:0.578) -- ++(150:0.578)-- ++(90:0.578) -- ++(150:0.578)-- ++(90:0.578) -- ++(150:0.578)-- ++(90:0.578) -- ++(150:0.578)-- ++(90:0.578) -- ++(150:0.578)
				--++(90:0.578) -- ++(30:0.578)--++(90:0.578) -- ++(30:0.578)--++(90:0.578) -- ++(30:0.578)--++(90:0.578) -- ++(30:0.578)--++(90:0.578) -- ++(30:0.578)--++(90:0.578) -- ++(30:0.578)
				--++(-30:0.578) -- ++(30:0.578)--++(-30:0.578) -- ++(30:0.578)--++(-30:0.578) -- ++(30:0.578)--++(-30:0.578) -- ++(30:0.578)--++(-30:0.578) -- ++(30:0.578)--++(-30:0.578) 
				-- ++(-90:0.578)--++(-30:0.578) -- ++(-90:0.578)--++(-30:0.578) -- ++(-90:0.578)--++(-30:0.578) -- ++(-90:0.578)--++(-30:0.578) -- ++(-90:0.578)--++(-30:0.578) 
				--++(-90:0.578)--++(-150:0.578)--++(-90:0.578)--++(-150:0.578)--++(-90:0.578)--++(-150:0.578)--++(-90:0.578)--++(-150:0.578)--++(-90:0.578)--++(-150:0.578)--++(-90:0.578)--++(-150:0.578)
				--++(150:0.578)--++(210:0.578)--++(150:0.578)--++(210:0.578)--++(150:0.578)--++(210:0.578)--++(150:0.578)--++(210:0.578)--++(150:0.578)--++(210:0.578)--cycle;
				
				\begin{scope}[xshift=3 cm, yshift=5.18 cm]
					\draw[->, line width=1.1 pt, color=red!50!orange!93!black!50!white] 
					(-60:4.55) .. controls +(30:0.5) and +(-120:0.5) ..
					(-30:4.25) .. controls +(60:0.5) and +(-90:0.5) ..       
					(0:4.55) .. controls +(90:0.5) and +(-60:0.5) ..
					(30:4.25) .. controls +(120:0.5) and +(-30:0.5) ..		
					(60:4.55) .. controls +(150:0.5) and +(0:0.5) ..
					(90:4.25) .. controls +(180:0.5) and +(30:0.5) ..		
					(120:4.55) .. controls +(210:0.5) and +(60:0.5) ..
					(150:4.25) .. controls +(240:0.5) and +(90:0.5) ..		
					(180:4.55) .. controls +(270:0.5) and +(120:0.5) ..
					(210:4.25) .. controls +(300:0.5) and +(150:0.5) ..		
					(240:4.55) .. controls +(330:0.5) and +(180:0.5) ..
					(270:4.25) .. controls +(0:0.5) and +(210:0.5) ..
					(301:4.55);
				\end{scope}
				\node[]  at (2.9, 8){\large{${I_{\sigma}}=\frac{J}{4}$}};
				
				\spy [blue, size=2.7 cm] on (0.84, 0.55) in node [] at (14.3,5.0);
				
			\end{scope}	
			\begin{scope}
				\node[]  at (1.4,0.7){\large{${R}$}};
				\node[]  at (0.6, 0.7){\large{${L}$}};
			\end{scope}

			\begin{scope}[scale=0.35, xshift=18.2 cm, yshift=8.4 cm, scale=0.91]	
				\node[]  at (-7.5+0.2,-5.9) {\large{$b$}};
				\node[]  at (-5.5+0.2,-5.9) {\large{$a$}};
				\node[]  at (-6.6+0.2,-5.9+0.578+1.1) {\large{$c$}};
				
				\node[]  at (-3.5+0.2,-5.9) {\large{$x$}};
				\node[]  at (-1.5+0.2,-5.9) {\large{$y$}};
				\node[]  at (-4.6+0.2,-5.9+0.578+1.1) {\large{$z$}};
				\node[]  at (-2.6+0.2,-5.9+0.578+1.1) {\large{$w$}};

				\begin{scope}[xshift=-4cm,yshift=0.25 cm]
					\node[]  at (10.5,-4.5) {$\left\{\begin{array}{l}
							f^{\mathfrak{D}}_{ax}=-\frac{5}{36}J\\[3pt]
							f^{\mathfrak{D}}_{ay}=\frac{1}{9}J\\[3pt]
							f^{\mathfrak{D}}_{aw}=\frac{1}{18}J\\[3pt]
							f^{\mathfrak{D}}_{bx}=\frac{1}{9}J\\[3pt]
							f^{\mathfrak{D}}_{bz}=\frac{1}{18}J\\[3pt]
							f^{\mathfrak{D}}_{cx}=\frac{1}{18}J		
						\end{array}\right.$};
				\end{scope}
			\end{scope}	
		\end{scope}
	\end{tikzpicture}
	\caption{Summary of all the nonzero modular currents that contribute to the edge modular current $I_{\sigma}$.}
	\label{fig:edge_modular_current}
\end{figure}

To be concrete, let us take the following formal definition of the edge modular current:
\begin{equation}
	\begin{aligned}
		I_{\sigma} := f^{\mathfrak{D}}(L,R),
	\end{aligned}
\end{equation}
where $L$ and $R$ are subsystems of $\mathfrak{D}$ such that they are sufficiently large and that the triple point that $L$, $R$, and the remainder meets together is sufficiently far away from the edge; see Fig.~\ref{fig:triple_point}(b). 
The deformation of $L$ and $R$ in the bulk does not change the modular current $f^{\mathfrak{D}}(L,R)$ because the contribution from any pair of bulk sites vanishes. Deformation along the edge also does not change anything, because the modular current is conserved and local. This is why $f^{\mathfrak{D}}(L,R)$ represents a well-defined notion of current along the edge, independent of the fine-grained details on the shapes of $L$ and $R$.

In practice, the choice of $L$ and $R$ in Fig.~\ref{fig:edge_modular_current} (as small as radius 1) is already a valid choice.  The calculation of $I_{\sigma}$ can be done straightforwardly, again using the topological invariance of the modular commutator. All nonzero contributions are summarized in Fig.~\ref{fig:edge_modular_current}, leading to
\begin{equation}\label{eq:edge_1/4J}
		I_{\sigma} = \frac{1}{4}J.
\end{equation}
As it stands, the edge modular current is indeed insensitive to the choice of disk $\mathfrak{D}$,  consistent with our expectation.

Now, using a physical argument explained in \cite{Short}, we can calculate the value of $J$. The idea is to interpret the modular Hamiltonian of the disk as some local Hamiltonian. Viewed this way, the modular current is the energy current of this local Hamiltonian at a ``temperature'' of $T=1$. Comparing this expression with Eq.~\eqref{eq:edge_current}, we obtain
\begin{equation}
    J= \frac{\pi}{3} c_-.
\end{equation}

\subsection{Edge modular current, with a bulk anyon}
Consider a disk containing an anyon, whose quantum state is $\sigma_{\mathfrak{D}}^{[a]}$ (instead of the ground state $\sigma$); see Fig.~\ref{fig:anyon_edge_current} for an illustration. It is expected that the presence of an anyon can modify the entanglement spectrum of the disk.\footnote{This modification of entanglement spectrum must happen for non-Abelian anyons whose quantum dimension $d_a >1$. This can be seen from the entropy difference Eq.~\eqref{eq:entropy_quantum_dimension_correction}.}
Nonetheless, as we shall see below, the edge modular current turns out to be well-defined and the value of this current is not affected by the presence of the anyon; see Eq.~(\ref{eq:a_and_1}). Recall that we have already established that the modular Hamiltonian of $\sigma_{\mathfrak{D}}^{[a]}$ is local; see Section~\ref{ss:anyon_disk}.

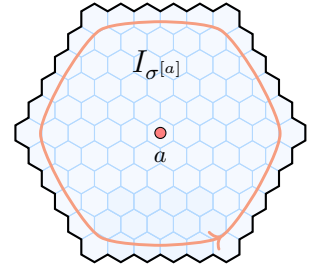
\begin{figure}[h]
	\centering
	\begin{tikzpicture}
		\begin{scope}
			\begin{scope}[scale=0.35]
				
				\diskshadded{1}{1}{6}{1}{blue!50!cyan!30!white}{blue!50!cyan!10!white,opacity=0.4};
				\diskshadded{2}{1}{3}{1}{blue!50!cyan!30!white}{blue!50!cyan!10!white,opacity=0.4};
				\diskshadded{4}{1}{5}{1}{blue!50!cyan!30!white}{blue!50!cyan!10!white,opacity=0.4};
				\diskshadded{1}{2}{7}{2}{blue!50!cyan!30!white}{blue!50!cyan!10!white,opacity=0.4};
				\diskshadded{2}{2}{3}{2}{blue!50!cyan!30!white}{blue!50!cyan!10!white,opacity=0.4};
				\diskshadded{4}{2}{6}{2}{blue!50!cyan!30!white}{blue!50!cyan!10!white,opacity=0.4};
				\diskshadded{1}{3}{8}{3}{blue!50!cyan!30!white}{blue!50!cyan!10!white,opacity=0.4};
				\diskshadded{3}{3}{4}{3}{blue!50!cyan!30!white}{blue!50!cyan!10!white,opacity=0.4};
				\diskshadded{5}{3}{6}{3}{blue!50!cyan!30!white}{blue!50!cyan!10!white,opacity=0.4};
				\diskshadded{1}{4}{9}{4}{blue!50!cyan!30!white}{blue!50!cyan!10!white,opacity=0.4};
				\diskshadded{1}{5}{10}{5}{blue!50!cyan!30!white}{blue!50!cyan!10!white,opacity=0.4};
				\diskshadded{1}{6}{11}{6}{blue!50!cyan!30!white}{blue!50!cyan!10!white,opacity=0.4};
				
				\diskshadded{2}{7}{11}{7}{blue!50!cyan!30!white}{blue!50!cyan!10!white,opacity=0.4};
				\diskshadded{3}{8}{11}{8}{blue!50!cyan!30!white}{blue!50!cyan!10!white,opacity=0.4};
				\diskshadded{4}{9}{11}{9}{blue!50!cyan!30!white}{blue!50!cyan!10!white,opacity=0.4};
				\diskshadded{5}{10}{11}{10}{blue!50!cyan!30!white}{blue!50!cyan!10!white,opacity=0.4};
				\diskshadded{6}{11}{11}{11}{blue!50!cyan!30!white}{blue!50!cyan!10!white,opacity=0.4};

				\draw[line width=0.8 pt] (30:0.578) -- (90:0.578) 
				-- ++(90:0.578) -- ++(150:0.578)-- ++(90:0.578) -- ++(150:0.578)-- ++(90:0.578) -- ++(150:0.578)-- ++(90:0.578) -- ++(150:0.578)-- ++(90:0.578) -- ++(150:0.578)
				--++(90:0.578) -- ++(30:0.578)--++(90:0.578) -- ++(30:0.578)--++(90:0.578) -- ++(30:0.578)--++(90:0.578) -- ++(30:0.578)--++(90:0.578) -- ++(30:0.578)--++(90:0.578) -- ++(30:0.578)
				--++(-30:0.578) -- ++(30:0.578)--++(-30:0.578) -- ++(30:0.578)--++(-30:0.578) -- ++(30:0.578)--++(-30:0.578) -- ++(30:0.578)--++(-30:0.578) -- ++(30:0.578)--++(-30:0.578) 
				-- ++(-90:0.578)--++(-30:0.578) -- ++(-90:0.578)--++(-30:0.578) -- ++(-90:0.578)--++(-30:0.578) -- ++(-90:0.578)--++(-30:0.578) -- ++(-90:0.578)--++(-30:0.578) 
				--++(-90:0.578)--++(-150:0.578)--++(-90:0.578)--++(-150:0.578)--++(-90:0.578)--++(-150:0.578)--++(-90:0.578)--++(-150:0.578)--++(-90:0.578)--++(-150:0.578)--++(-90:0.578)--++(-150:0.578)
				--++(150:0.578)--++(210:0.578)--++(150:0.578)--++(210:0.578)--++(150:0.578)--++(210:0.578)--++(150:0.578)--++(210:0.578)--++(150:0.578)--++(210:0.578)--cycle;
				
				\draw[fill=red!50] (30:0.578*6)++(90:0.578*6) circle (0.21);
				\node[]  at (3.003, 4.3){$a$};
				\begin{scope}[xshift=3 cm, yshift=5.18 cm]
					\draw[->, line width=1.1 pt, color=red!50!orange!93!black!50!white] 
					(-60:4.55) .. controls +(30:0.5) and +(-120:0.5) ..
					(-30:4.25) .. controls +(60:0.5) and +(-90:0.5) ..       
					(0:4.55) .. controls +(90:0.5) and +(-60:0.5) ..
					(30:4.25) .. controls +(120:0.5) and +(-30:0.5) ..		
					(60:4.55) .. controls +(150:0.5) and +(0:0.5) ..
					(90:4.25) .. controls +(180:0.5) and +(30:0.5) ..		
					(120:4.55) .. controls +(210:0.5) and +(60:0.5) ..
					(150:4.25) .. controls +(240:0.5) and +(90:0.5) ..		
					(180:4.55) .. controls +(270:0.5) and +(120:0.5) ..
					(210:4.25) .. controls +(300:0.5) and +(150:0.5) ..		
					(240:4.55) .. controls +(330:0.5) and +(180:0.5) ..
					(270:4.25) .. controls +(0:0.5) and +(210:0.5) ..
					(301:4.55);
				\end{scope}				
				\node[]  at (2.9, 7.8){\large{$I_{\sigma^{[a]}}$}};
			\end{scope}
		\end{scope}
		\end{tikzpicture}
	\caption{A disk subsystem $\mathfrak{D}\subset \Lambda$ containing an anyon, which is in a quantum state $\sigma^{[a]}_{\mathfrak{D}}$. The edge modular current is denoted as $I_{\sigma^{[a]}}$.}\label{fig:anyon_edge_current}
	\end{figure}

The local decomposition of $\sigma^{[a]}_{\mathfrak{D}}$ and  $\sigma_{\mathfrak{D}}$ only differs in the bulk, in the vicinity of the anyon. This is because they have identical reduced density matrices on any disk $D\subset \mathfrak{D}$ that does not contain the anyon. Thus, the commutators of terms (of the form $\langle i [K^{\mathfrak{D}}_v, K^{\mathfrak{D}}_u]\rangle$) near the boundary of $\mathfrak{D}$ are unaffected. Therefore, it is natural to speculate that the bulk anyon does not change the edge modular current, \emph{i.e.,}
\begin{equation}\label{eq:a_and_1}
	I_{\sigma^{[a]}} 
	{=} I_{\sigma}. 
\end{equation}
As we shall see below, Eq.~(\ref{eq:a_and_1}) is indeed true. However, the argument, as it stands, is incomplete. The rest of this section is dedicated to completing this argument.

The missing piece of the argument is the justification of the fact that the edge modular current with respect to  $\sigma^{[a]}_{\mathfrak{D}}$ (denoted as  $I_{\sigma^{[a]}}$) is well-defined. Recall that the definition of $I_{\sigma}$ rests on three simple properties of the modular current $f_{vu}^{\mathfrak{D}}$ of the ground state $\sigma$: locality, conservation of the current, and the fact that it vanishes in the bulk. We need to justify the same properties for the modular current of the state $\sigma^{[a]}_{\mathfrak{D}}$. Both the locality and the conservation follow straightforwardly from the analysis in Section~\ref{ss:anyon_disk}. However, that the current vanishes in the bulk is less obvious and requires a further analysis.

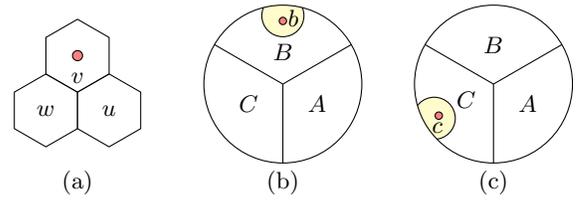
\begin{figure}[h]
	\centering
	\begin{tikzpicture}
		\begin{scope}[scale=0.35]
			\begin{scope}[xshift=-9 cm, yshift=-1 cm, scale=2.4]
				\coloredboxhx{white,opacity=0.5}{0}{0}{1}{0};
				\coloredboxhx{white,opacity=0.5}{1}{1}{1}{1};
				\node at (0,0) {$w$};
				\node at (1,0) {$u$};
				\node at (0.5,0.51) {$v$};
				\draw[fill=red!50] (30:0.578*1)++(90:0.578*1) circle (0.08);
		    \end{scope}
	    \node (a) at (-9+0.5*2.4, -3.75) {(a)};
		
			\begin{scope}
				\begin{scope}
					\clip(0,0) circle (3);
					\fill[yellow, opacity=0.2] (90:2.6) circle (0.8);
					\draw(90:2.6) circle (0.8);
					\draw[fill=red!50] (90:2.4) circle (0.14);
				\end{scope}
				\draw[] (0,0) circle (3);
				\node at (90:1.2) {$B$};
				\node at (-30:1.5) {$A$};
				\node at (-150:1.5) {$C$};
				\draw (0:0)--(30:3);
				\draw (0:0)--(150:3);
				\draw (0:0)--(-90:3);
				\node (a) at (90-9:2.55) {$b$};
				\node (a) at (-90:3.75) {(b)};
			\end{scope}	
		
		\begin{scope}[xshift=8 cm]
			\begin{scope}
				\clip(0,0) circle (3);
				\fill[yellow, opacity=0.2] (210:2.6) circle (0.8);
				\draw(210:2.6) circle (0.8);
				\draw[fill=red!50] (180+30:2.4) circle (0.14);
			\end{scope}
			\draw[] (0,0) circle (3);
			\node at (90:1.5) {$B$};
			\node at (-30:1.5) {$A$};
			\node at (-150:1.2) {$C$};
			\draw (0:0)--(30:3);
			\draw (0:0)--(150:3);
			\draw (0:0)--(-90:3);
			\node (a) at (180+38:2.7) {$c$};
			\node (a) at (-90:3.75) {(c)};
		\end{scope}	
		\end{scope}	
	\end{tikzpicture}
	\caption{An anyon $a$ is located in an appropriate region of a tripartition of a disk. (a) The anyon is contained in $v$. (b) The anyon is contained in $b$ and $b \subset B$. (b) The anyon is contained in $c$ and $c\subset C$.}\label{fig:anyon_deformation}
\end{figure}

With an appropriate coarse-graining, we can always ensure that the anyon resides strictly inside a single unit cell. Therefore, it suffices to consider two types of modular commutators, $J(u,v,w)_{\sigma^{[a]}}$ and $J(w,u,v)_{\sigma^{[a]}}$ with respect to the partitions in Fig.~\ref{fig:anyon_deformation}(a). The argument for the topological invariance of the modular commutator in Section~\ref{sec:topological_invariance} applies here as well, as long as the deformation occurs sufficiently far away from the anyon. This is because the requisite conditional independence relation was derived from \textbf{A1} (see Eq.~\eqref{eq:ssa_saturation}), a \emph{local} condition on $\sigma$. Therefore, all we need to show is that $J(A,B,C)_{\sigma^{[a]}}=J$ for the partition in Fig.~\ref{fig:anyon_deformation} (b) and (c) respectively. 

If we remove the anyon by removing region $b$ in Fig.~\ref{fig:anyon_deformation}(b), or removing region $c$ from (c), then the tripartition of the disk contains no anyon. Therefore, $J(A,B\setminus b, C)= J$ for (b) and $J(A,B,C\setminus c)=J$ for (c), from the topological invariance of $J$. Therefore, the remaining task is to show
\begin{itemize}
	\item $J(A,B, C)_{\sigma^{[a]}}= J(A,B\setminus b, C)_{\sigma^{[a]}}$ for the partition in  Fig.~\ref{fig:anyon_deformation}(b).
	\item $J(A,B, C)_{\sigma^{[a]}}= J(A,B, C\setminus c)_{\sigma^{[a]}}$ for the partition in  Fig.~\ref{fig:anyon_deformation}(c).
\end{itemize}
We can prove both statements, using an argument similar to the one used in Section~\ref{sec:topological_invariance}. A small difference is that the requisite conditional independence condition does not follow immediately from \textbf{A1}. Instead, we need to show $I(b:A \, \vert \, B\setminus b)=I(b:C \, \vert \, B\setminus b)=0$ for the partition in Fig.~\ref{fig:anyon_deformation}(b) and  $I(c:B \, \vert \, C\setminus c)=0$ for the partition in Fig.~\ref{fig:anyon_deformation}(c). Happily, both of these conditions follow straightforwardly from known results~\cite{Kitaev2006,SKK2019}. Thus, the modular current for the state $\sigma^{[a]}_{\mathfrak{D}}$ vanishes in the bulk.
 This completes the derivation of Eq.~(\ref{eq:a_and_1}).

The consequence of this analysis is the physical prediction that a single anyon in the bulk does not modify the magnitude of the edge energy current, when the temperature is low compared to the bulk energy gap and high enough such that the boundary length is large compared to the thermal correlation length. (In other words, formula~(\ref{eq:edge_current}) still holds when there is an anyon in the bulk.) This prediction is consistent with existing literature; see Appendix~D.2 of Ref.~\cite{Kitaev2006a}.

\section{Gapped domain wall}
\label{sec:gapped_domain_wall}
In this Section, we show that two systems which are joined via a gapped domain wall has the same value of $J$. This is a nontrivial evidence that $J$ is indeed proportional to the chiral central charge.

The notion of gapped domain wall we use is based on the entanglement bootstrap program~\cite{Shi2021,Shi2021a}. In the entanglement bootstrap program, one begins with a set of axioms on local patches of the wave function and deduces the logical consequence of these axioms. These axioms are summarized in Fig.~\ref{fig:assumptions}.  Specifically, let $\sigma$ be some reference state, which can be again thought as the ground state of the many-body system. We assume that $(S_C + S_{BC} - S_B)_{\sigma}=0$ (red) and $(S_{BC} + S_{CD} - S_B - S_D)_{\sigma}=0$ (green), both in the bulk and on the domain wall, over arbitrarily large regions that can be smoothly deformed from the shown configurations. The subsystems are allowed to be deformed as long as the boundaries between $B$ and $D$ do not cross the domain wall. We shall refer to the conditions in the bulk and on the wall as the bulk-axioms and the wall-axioms, respectively.  

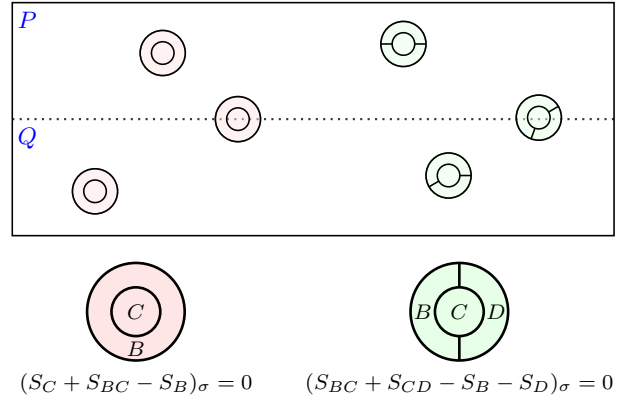
\begin{figure}[h]
	\centering
	\begin{tikzpicture}
		\draw[line width=0.6 pt] (0,-1.55) rectangle (8,1.55);
		\draw[dotted, line width=0.8 pt, opacity=0.8] (0,0) --(8,0);
		\node[right, below] (P) at (0.2, 1.55)  {{\color{blue}$P$}};
		\node[right, below] (Q) at (0.2, 0)  {{\color{blue}$Q$}};
		\axiomone{1.1}{-0.96}
		\axiomone{2}{0.88}
		\axiomone{3}{0}
		\axiomtwoangled{5.2}{1}{0}{180}
		\axiomtwoangled{5.8}{-0.75}{0}{210}
		\axiomtwoangled{7}{0.02}{30}{250}
	\end{tikzpicture}
	\vspace{0.3cm}
	
	\begin{tikzpicture}
		\begin{scope}[scale=0.865]
			\draw[fill=red!10!white, line width=1pt] (0,0) circle (0.75);
			\draw[fill=red!10!white, line width=1pt] (0,0) circle (0.375);
			\node[] (A) at (0,-1.1) {\footnotesize $(S_C + S_{BC} - S_B)_{\sigma}=0$};
			\node[] (B) at (0,-0.5625) {\footnotesize{$B$}};
			\node[] (C) at (0,0) {\footnotesize{$C$}};
		\end{scope}
		
		\begin{scope}[xshift=4.3cm,scale=0.865]
			\draw[fill=green!10!white, line width=1pt] (0,0) circle (0.75);
			\draw[fill=green!10!white, line width=1pt] (0,0) circle (0.375);
			\node[] (A) at (0,-1.1) {\footnotesize $(S_{BC} + S_{CD} - S_B - S_D)_{\sigma}=0$};
			\node[] (B) at (-0.5625,0) {\footnotesize{$B$}};
			\node[] (D) at (0.5625,0) {\footnotesize{$D$}};
			\node[] (C) at (0,0) {\footnotesize{$C$}};
			\draw[line width=1pt] (0,0.375) -- (0,0.75);
			\draw[line width=1pt] (0,-0.375)--(0,-0.75);
		\end{scope}
	\end{tikzpicture}
	\caption{Summary of the domain wall axioms.}
	\label{fig:assumptions}
\end{figure}

\begin{figure}[h]
	\centering
	\begin{tikzpicture}
		\begin{scope}[scale=0.8]
			\draw[line width=0.6 pt] (0,-3) rectangle (4.6,3);
			\draw[dotted,line width=0.8 pt, opacity=0.8] (0,0) --(4.6,0);
			\node at (2.3, -3.4)  {(a)};
			\begin{scope}[scale=0.92,xshift=0.1 cm]
				\node at (0.2,2.95) {\color{blue}$P$};
				\node at (0.2, -0.35)  {{\color{blue}$Q$}};
				\begin{scope}[xshift=2.2 cm, yshift=0 cm]	
					\draw	(-1,0)+(-50:2.5) arc (-50:50:2.5);
					\draw	(1,0)+(130:2.5) arc (130:230:2.5);
					\draw	(-0.2,0)+(-50:2.2) arc (-50:50:2.2);
					
					\node at (-10:1.75) {$A'$};
					\node at (20:0.7) {$D$};
					
					\node at (0,2.6) {$E$};
				\end{scope}	
				
				\begin{scope}[xshift=2.2 cm, yshift=1.3 cm, xscale=0.45, yscale=0.3]	
					\draw[fill=white] (0,0) circle (3);
					\node at (-90:1.5) {$B$};
					\node at (30:1.5) {$A$};
					\node at (150:1.5) {$C$};
					\draw (0:0)--(-30:3);
					\draw (0:0)--(-150:3);
					\draw (0:0)--(90:3);	
				\end{scope}	
				
				\begin{scope}[xshift=2.2 cm, yshift=-1.3 cm, xscale=0.45, yscale=0.3]	
					\draw[fill=white] (0,0) circle (3);
					\node at (90:1.5) {$B$};
					\node at (-30:1.5) {$A$};
					\node at (-150:1.5) {$C$};
					\draw (0:0)--(30:3);
					\draw (0:0)--(150:3);
					\draw (0:0)--(-90:3);	
				\end{scope}	
			\end{scope}	
		\end{scope}	
		
		\begin{scope}[scale=0.8, xshift=5.0 cm]
			\draw[line width=0.6 pt] (0,-3) rectangle (4.6,3);
			\draw[dotted,line width=0.8 pt, opacity=0.8] (0,0) --(4.6,0);
			\node at (2.3, -3.4)  {(b)};
			\begin{scope}[scale=0.92,xshift=0.1 cm]
				\node at (0.2,2.95) {\color{blue}$P$};
				\node at (0.2, -0.35)  {{\color{blue}$Q$}};
				\begin{scope}[xshift=2.2 cm, yshift=0 cm]	
					\draw	(-1,0)+(-50:2.5) arc (-50:50:2.5);
					\draw	(1,0)+(130:2.5) arc (130:230:2.5);
					\draw	(-0.2,0)+(-50:2.2) arc (-50:50:2.2);
					
					\draw	(0,-1)--(0, 1);
					
					\node at (-10:1.75) {$A'$};
					\node at (15:0.9) {$D_2$};
					\node at (165:0.9) {$D_1$};
					
					\node at (0,2.6) {$E$};
				\end{scope}	
				
				\begin{scope}[xshift=2.2 cm, yshift=1.3 cm, xscale=0.45, yscale=0.3]	
					\draw[fill=white] (0,0) circle (3);
					\node at (-90:1.5) {$B$};
					\node at (30:1.5) {$A$};
					\node at (150:1.5) {$C$};
					\draw (0:0)--(-30:3);
					\draw (0:0)--(-150:3);
					\draw (0:0)--(90:3);	
				\end{scope}	
				
				\begin{scope}[xshift=2.2 cm, yshift=-1.3 cm, xscale=0.45, yscale=0.3]	
					\draw[fill=white] (0,0) circle (3);
					\node at (90:1.5) {$B$};
					\node at (-30:1.5) {$A$};
					\node at (-150:1.5) {$C$};
					\draw (0:0)--(30:3);
					\draw (0:0)--(150:3);
					\draw (0:0)--(-90:3);	
				\end{scope}	
				
			\end{scope}
		\end{scope}	
	\end{tikzpicture}
	\caption{Passing through the domain wall.}
	\label{fig:passing}
\end{figure}
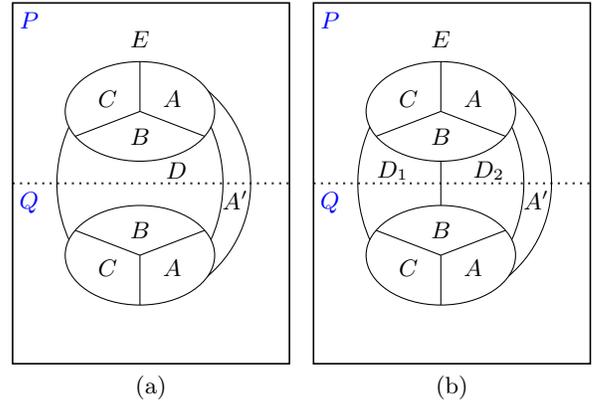

The main result of this Section is the following theorem. 

\begin{theorem}\label{wall_thm}
	The modular commutator calculated from gapped phases $P$ and $Q$ are identical,
	\begin{equation}
		J^P=J^Q,
	\end{equation}
	if the two phases $P$ and $Q$ are separated by a gapped domain wall.
\end{theorem}

\begin{proof}
	Proving the statement $J^P=J^Q$ is equivalent to proving the statement that $J(A,B,C)_{\sigma}=0$ for the partition shown in Fig.~\ref{fig:passing}(a). This is because $J(A,B,C)$ is additive under a tensor product; see Eq.~(\ref{eq:additive}). (Here the relevant tensor product is the tensor product structure of $\sigma$ on a disk in the $P$ phase and its ``mirror image" in the $Q$ phase.) By the axiom on the bottom left corner of Fig.~\ref{fig:assumptions} (also known as {\bf A0} in Ref.~\cite{SKK2019}), any bulk region within $P$ and its counterpart in $Q$ are decoupled from each other.
	
	Let us consider a reference state $\sigma$ on a large enough disk containing $AA'BCD$ and an extra layer around it. Let $\vert \psi_{AA'BCDE}\rangle$ be the purification of the reference state; here $E$ is the complement of $AA'BCD$ (with the purifying space included). It follows from the bulk-axiom {\bf A1} that
	\begin{equation}
		I(B:E\vert C)_{\vert \psi\rangle }= I(B:A'\vert A)_{\vert \psi\rangle }=0.
	\end{equation}
	Using the identity $K_{BCE} = K_{BC} + K_{CE} - K_C$, we obtain $J(A,B,C)_{\vert\psi\rangle}= J(A,B,CE)_{\vert\psi\rangle}$. Similarly, using $K_{AA'B}=K_{AA'} + K_{AB} - K_A$, we obtain $J(A,B,C)_{\vert\psi\rangle}=J(AA',B,CE)_{\vert\psi\rangle}$.
	Next, we deform the operators according to $K_X \vert \psi_{XY}\rangle = K_Y \vert \psi_{XY} \rangle$. We find that
	\begin{equation}
		\langle [K_{AA'B},K_{BCE}]\rangle = \langle [K_{CDE},K_{AA'D}] \rangle.
	\end{equation}
	We further divide $D$ into $D_1D_2$, as is shown in Fig.~\ref{fig:passing}(b). From the wall-axiom {\bf A1}, one derives
	\begin{equation}
		I(D_2:CE\vert D_1)_{\vert \psi\rangle}=I(D_1: AA' \vert D_2)_{\vert \psi\rangle}=0 \label{eq:Markov_wall}
	\end{equation} 
	It follows that 
	\begin{equation}
		\begin{aligned}
			&\langle [K_{CDE},K_{AA'D}] \rangle \\
			&= \langle[K_{CED_1} + K_{D_1D_2} -K_{D_1}, K_{AA'D_2} + K_{D_1D_2} -K_{D_2}] \rangle\\
			&= \langle [K_{CE D_1}, K_{D_1D_2}] + [K_{D_1D_2}, K_{AA'D_2}] \rangle\\
			&=0.
		\end{aligned}
	\end{equation}
	The second line follows from the rewriting of the modular Hamiltonian using the Markov chain conditions Eq.~(\ref{eq:Markov_wall}). The third line is arrived after dropping non-overlapping terms and applying $\langle [K_{A}, K_{AB}]\rangle=0$. The last line follows from the vanishing of modular commutators due to the Markov chain condition Eq.~(\ref{eq:Markov_wall}).

	In conclusion, $J(A,B,C)_{\sigma}=0$ for the partition in Fig.~\ref{fig:passing}, which implies that $J^Q=J^P$.
\end{proof}

	To see the key point of the proof, we encourage the readers to repeat the proof for a special topology explained below, for which the topology of the regions may be easier to grasp. Suppose that the reference state is on a sphere and it is pure, denoted as $\vert \psi\rangle$. The northern hemisphere ($P$ phase) and the southern hemisphere ($Q$ phase) are separated by a domain wall at the equator. 
	
	We can formally view the sphere as a double-layered disk, with the two layers connected along the boundary of the disk. (The northern and the southern hemisphere can be thought as the first and the second layer, respectively.) This way, we can view the state on the sphere as a state on a disk with a gapped boundary. We want to show that $J(A,B,C)_{\vert \psi\rangle}=0$ for the partition of this disk shown in Fig.~\ref{fig:gapped_boundary_disk}.	
	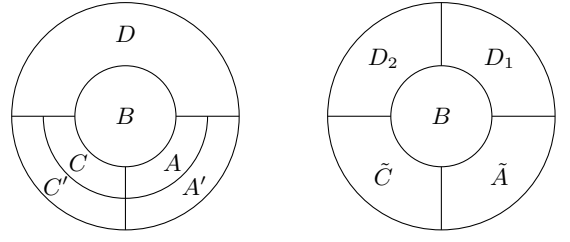
\begin{figure}[h]
		\centering
		\begin{tikzpicture}
			\begin{scope}[scale=0.84]
				\draw[] (0,0) circle (1.8);
				\draw[] (0,0) circle (0.8); 
				\draw[] (0:0.8) -- (0:1.8); 
				\draw[] (-90:0.8) -- (-90:1.8); 
				\draw[] (180:0.8) -- (180:1.8); 
				\draw[] (-180:1.3) arc (-180:0:1.3);
				
				\node at (0:0) {$B$};
				\node at (-45:1.05) {$A$};
				\node at (-45:1.55) {$A'$};
				\node at (-135:1.05) {$C$};
				\node at (-135:1.55) {$C'$};
				\node at (90:1.3) {$D$};
			\end{scope}	
			
			\begin{scope}[scale=0.84, xshift=5 cm]
				\begin{scope}[scale=1]
					\draw[] (0,0) circle (1.8);
					\draw[] (0,0) circle (0.8); 
					\draw[] (0:0.8) -- (0:1.8); 
					\draw[] (-90:0.8) -- (-90:1.8); 
					\draw[] (180:0.8) -- (180:1.8); 
					\draw[] (90:0.8) -- (90:1.8);
					
					\node at (0:0) {$B$};
					\node at (-45:1.3) {$\tilde{A}$};
					\node at (-135:1.3) {$\tilde{C}$};
					\node at (45:1.3) {$D_1$};
					\node at (135:1.3) {$D_2$};
				\end{scope}	
			\end{scope}	
		\end{tikzpicture}
		\caption{A physical system on a disk with a gapped boundary and the partitions of the disk. $\tilde{A}$ on the right is $AA'$ on the left; $\tilde{C}$ on the right is $CC'$ on the left; $D_1D_2$ on the right is $D$ on the left.}
		\label{fig:gapped_boundary_disk}
	\end{figure}

    It is easy to see that the regions in Fig.~\ref{fig:gapped_boundary_disk} are identical to those in Fig.~\ref{fig:passing} once we relabel $C'$ as $E$. Every step of the proof goes through.
    
	In particular, $J(\tilde{A},B,\tilde{C})_{\vert \psi \rangle }=J(A,B,C)_{\vert \psi\rangle}$ due to $I(B:A'\vert A)_{\vert \psi\rangle}=I(B:C'\vert C)_{\vert \psi\rangle}=0$. Then
	\begin{equation}
		\begin{aligned}
		&	\langle [K_{\tilde{A}B}, K_{B\tilde{C}}]\rangle_{\vert \psi\rangle} \\
		&= \langle [K_{\tilde{C}D}, K_{\tilde{A}D}] \rangle_{\vert \psi\rangle}\\
		&= \langle[K_{\tilde{C}D_2}+K_{D_1D_2 } - K_{D_2}, K_{\tilde{A}D_1}+K_{D_1D_2} - K_{D_1} ] \rangle_{\vert \psi\rangle}\\
		&=0.
		\end{aligned}	
	\end{equation}
    Thus $J(A,B,C)_{\vert \psi\rangle}=0$ for this gapped boundary problem. This implies that $J^Q=J^P$ under the aforementioned special topology.

The proof making use of Fig.~\ref{fig:passing} is based on the same idea but it is more flexible: it does not put any requirement of the global topology of the physical system and there is no need to assume a pure reference state.
	
\begin{corollary}\label{wall_coro}
	The modular commutator, for the partition in Fig.~\ref{fig:abc}, has $J(A,B,C)=0$ for any 2D reference state that admits a gapped boundary separating it and the vacuum.
\end{corollary}

The physical implication of theorem~\ref{wall_thm} and corollary~\ref{wall_coro} is that the equality of chiral central charge,
\begin{equation}
	c_-(P) = c_-(Q),
\end{equation} is a necessary condition for the two gapped 2D phases to be separated by a gapped domain wall. (At least, this has to be true if the phases can be described by the entanglement bootstrap approach.) Thus, the chiral central charge for a gapped phase describable by the entanglement bootstrap approach must have $c_-=0$ if it admits a gapped boundary.

It is worth noting that $c_-=0$ is a necessary condition for a gapped boundary to exist. However, it is not a sufficient condition. See \cite{Levin2013} for an example. In a broader context, these obstructions come from the \emph{higher central charges} \cite{Ng2018,Ng2020,Kaidi2021}. It remains an interesting question whether it is possible to determine a necessary and sufficient condition for gapped boundary and domain walls to exist just from a ground state wave function.

\section{Numerical evidence}
\label{sec:numerical_evidence}
In this Section, we provide a numerical evidence for our formula. We consider a  model ground state wave function of an interacting quantum spin system on a lattice, whose underlying phase is known and well-studied. This is a model wave function that describes a chiral topological order with semion statistics, introduced by Nielsen \emph{et al}~\cite{Nielsen2012}. This wave function can be thought as a discrete version of the $\nu=\frac{1}{2}$ bosonic Laughlin state~\cite{Laughlin1983,Kalmeyer1987}.

While our focus lies in the chiral central charge, we also study the total quantum dimension $\mathcal{D}$ for comparison. That the latter quantity can be extracted from a topologically ordered ground state wave function is well-known~\cite{Kitaev2006,Levin2006}. For our system of interest, $c_-=1$ and $\mathcal{D}=\sqrt{2}$, leading to 
\begin{equation}
\begin{aligned}
J &= \pi/3 \approx 1.047,\\
\gamma &= \ln \sqrt{2}\approx 0.347.
\end{aligned}
\end{equation}
Here $\gamma$ is the  topological entanglement entropy~\eqref{eq:teedef} calculated using the partition in Fig.~\ref{fig:abc}.

Here are the essential details of the wave function \cite{Nielsen2012} relevant to our calculation. Consider a pair $(N, \{z_j\}_{j=1}^N)$, where $N$ is an even number and $z_j\in \mathbb{C}$; all $z_j$ are distinct but otherwise arbitrary. For each position $z_j$ we assign a qubit; the total Hilbert space, being a tensor product of these $2$-dimensional ones, is $2^N$-dimensional. 
The (un-normalized) wave function is defined as 
\begin{equation}
\vert \Psi(N, \{z_j\}_{j=1}^N) \rangle = \sum_{\{s_i\}_{i=1}^N} c(\{s_i\}_{i=1}^N) \vert \{s_i\}_{i=1}^N \rangle.
\end{equation}
Here $s_i= \pm 1$ and the complex coefficients are given by 
\begin{equation}
c(\{s_i\}_{i=1}^N)= \delta_{\mathbf{s}} \prod_{n<m}^N (z_n-z_m)^{\frac{1}{2} s_n s_m},
\end{equation}
where $\delta_{\mathbf{s}}=1$ for $\sum_i s_i =0$ and $\delta_{\mathbf{s}}=0$ otherwise. (Compared to Ref.~\cite{Nielsen2012}, we have chosen $\alpha=\frac{1}{2}$ for semion and a trivial phase factor of $\chi_{p,s_p}=1$.)

In order to carry out the calculation of $J$ and $\gamma$, we put the model wave function on a sphere.\footnote{We chose the sphere geometry to minimize the finite-size effects. In principle, one could choose a regular lattice on a plane. However, the number of data points that can be obtained this way is too small to analyze.} This can be done via the stereographic projection. Namely, we assign $(\theta_j, \phi_j)$ to $z_j$ according to
\begin{equation}
z_j= \frac{\sin\theta_j}{1+ \cos\theta_j} \exp\left(\phi_j\right),
\end{equation}
where $\theta_j\in [0,\pi]$ and $\phi_j\in [0,2\pi)$. The corresponding unit vector of in 3-dimensional Euclidean space is $\vec{n}_j= (\sin\theta_j \cos\phi_j, \sin\theta_j \sin\phi_j,\cos\theta_j)$. We consider an ``evenly" distributed set of points and take $N\to \infty$ limit via extrapolation.

Concretely, our lattice is defined in terms of the golden angle $\varphi\equiv  2\pi \cdot \left( 1-\frac{\sqrt{5}-1}{2} \right) \approx 137.5^{\circ}$~\cite{naylor2002golden,vogel1979better}. We define a class of lattice $\Lambda(N,\varphi)$ on a unit sphere, for integer $N\ge 2$, according to the coordinates of its lattice sites
\begin{equation}
(\theta_j, \phi_j)= \left( \arccos \left(1-\frac{2(j-1)}{N-1}\right), (j-1)\varphi \right),
\end{equation}
where $j=1,2,\cdots, N$.
See Fig.~\ref{fig:spiral_24} for an illustration of $N=24$. The golden angle automatically makes the points fairly uniform for large $N$.
\begin{figure}[h]
	\includegraphics[scale=0.34]{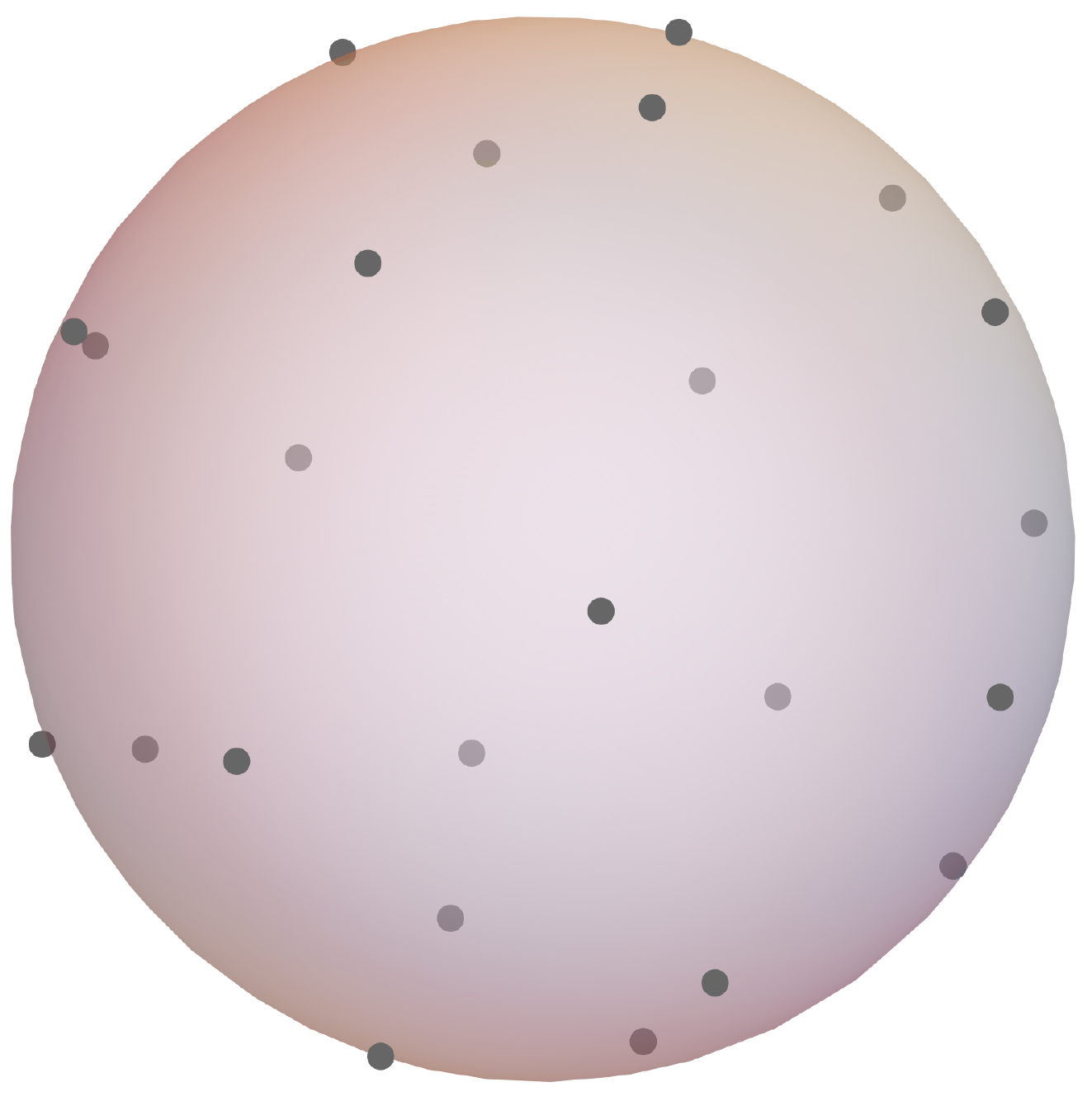}
	\caption{The lattice $\Lambda(24,\varphi)$. The 24 points are distributed fairly even on a sphere.}\label{fig:spiral_24}
\end{figure}

The advantage of our choice of lattice is that one can generate a fairly evenly distributed set of points on a sphere with a modest effort. While these points are not perfectly evenly distributed, the calculated value of $J$ should be insensitive to these microscopic details because our derivation of $J=\frac{\pi}{3}c_-$ did not depend on them. Indeed, our numerical analysis confirms this expectation.

Specifically, we computed the semion wave function for the lattice $\Lambda(N,\varphi)$ for $N\in \{12,14,16,18,20,22,24,26 \}$. For each $N$, we partition the sphere into four subsystems $A,B,C$ and $D$ such that they have a tetrahedral symmetry. Because the points on the sphere are not perfectly evenly distributed, the value of $J$ depends on which set of points belong to which of the subsystems. We thus average over different choices of subsystems over Haar-random rotation of the sphere. This way, we have obtained at least $1000$ randomly generated values of $(J,\gamma)$ for each $N\in \{12,14,16,18,20,22,24,26 \}$. (The largest system size of $N=26$ was limited to $1000$ samples, whereas the other system sizes ranged all the way up to $1.6\times 10^4$.) For the largest size we can probe ($N=26$), we find $\gamma \approx 0.381$ and $J \approx 0.964$, each deviating $11.5\%$ and $8.3\%$ from the theoretically predicted value, respectively. 

Intriguingly, the deviation monotonically decreases as $N$ increases; see Fig.~\ref{fig:fit}. Thus, by making an educated guess on the functional form of the correction, we can extrapolate the function to the $N\to \infty$ limit. We used two different families of ansatz:
\begin{equation}
\begin{aligned}
    f_p(x) &= ax^{-b} + c, \\
    f_e(x) &= ae^{-bx^{d}} + c,
\end{aligned}
\end{equation}
each corresponding to the ``power-law'' and ``exponential'' ansatz.  For $J$, the exponential ansatz provided a significantly better fit than the power-law ansatz. In particular, the value of $d$ in the range of $0.4<d <0.6$ provided comparable results, with the extrapolated value ranging between $1.067 \pm 0.013$ and $1.045\pm 0.009$. (We note that, outside of this range, the quality of the fit degraded significantly. In particular, for $d\geq 0.7$, we could not fit the function to the data due to convergence issue.) In contrast, for $\gamma$, the exponential ansatz did not provide a significantly better fit than the power-law based ansatz. For the exponential ansatz $f_e(x) = ae^{-bx^{1/2}}+ c$, the extrapolated value was $0.112\pm 0.062$, a significant deviation from $\ln \sqrt{2}\approx 0.347$. The best fit was produced by choosing $d=1.5$, yielding $0.3384\pm 0.0004$. For the power law ansatz, for $b=0.5$ and $b=1$, the extrapolated values were  $0.318\pm 0.005$ and $0.217\pm 0.006$, respectively. Thus, compared to the topological entanglement entropy, extrapolated value of $J$ is significantly more consistent with our theoretical prediction.

\begin{figure}[h]
	\includegraphics[width=0.95\columnwidth]{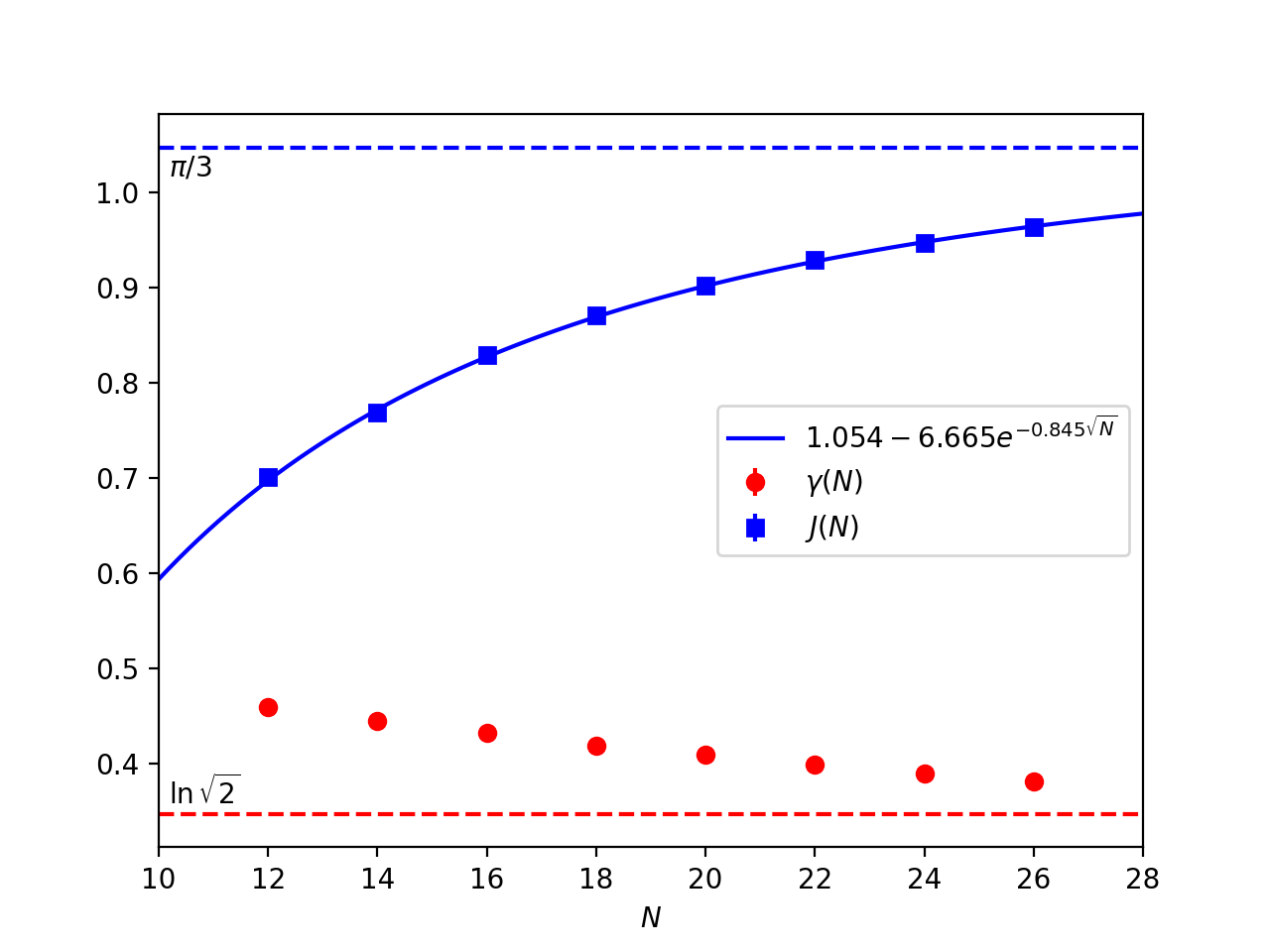}
	\includegraphics[width=0.95\columnwidth]{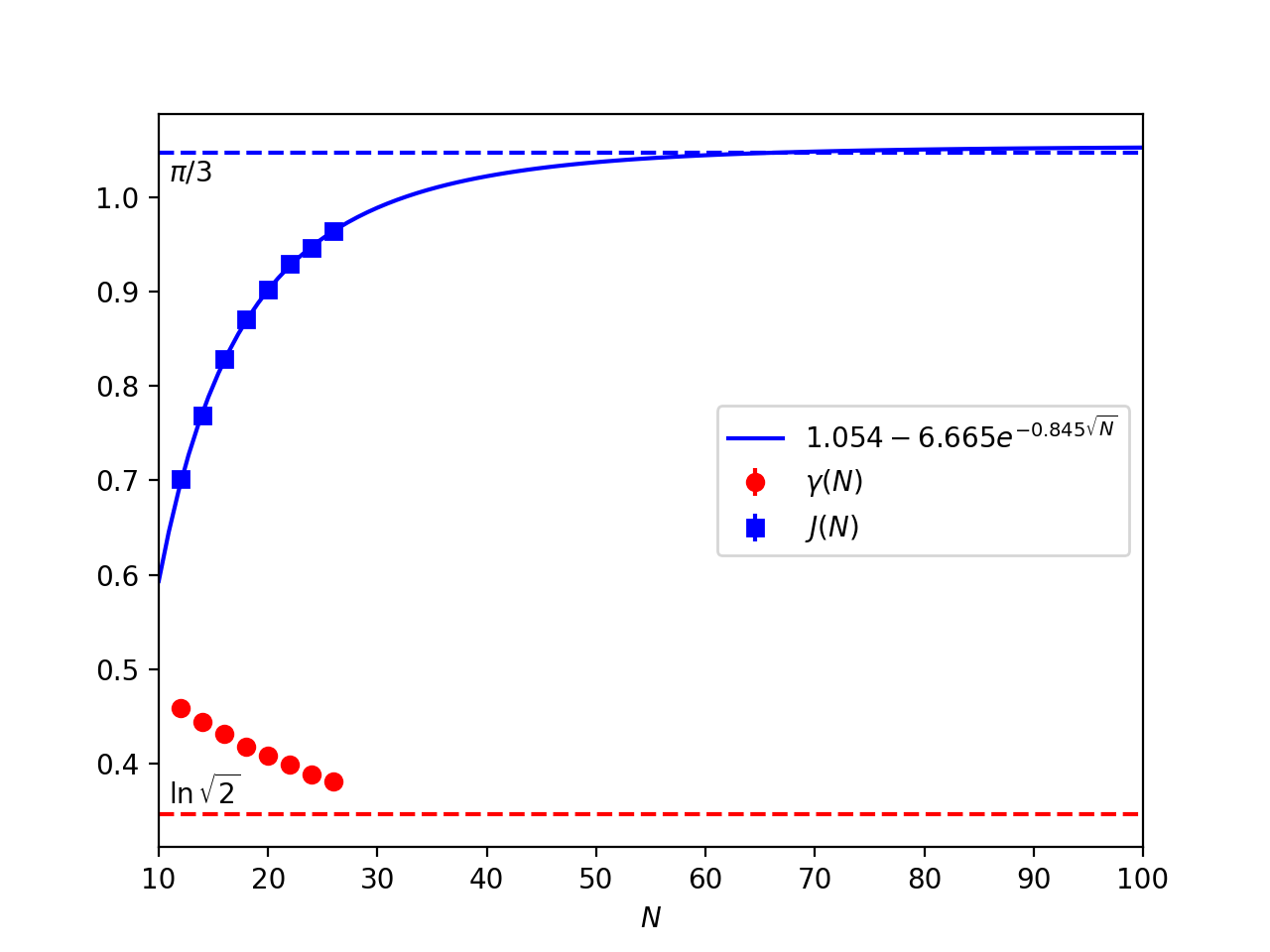}
	\caption{Mean values of $J(N)$ (blue) and $\gamma(N)$ (red). Both figures represent the same data, but the second figure plots the fitted function over larger $N$. Dashed lines represent the theoretically predicted values of $J=\pi/3$ and $\gamma=\ln \sqrt{2}$.}\label{fig:fit}
\end{figure}

While the precise dependence of $J$ on system size $N$ is unknown at this point, we can make an educated guess, based on the analysis above and the following physical argument. In the large $N$ limit, we expect the correlation length to be $\xi \propto 1/\sqrt{N}$, in the unit of the radius of the sphere. If the correction term decays exponentially with the inverse of the correlation length, we can expect the correction to decay exponentially in $N^{0.5}$. This meshes well with the fact that the fit for the exponential ansatz performed well in the $0.4<b<0.6$ range. Thus, we propose the following ansatz for $J$ for finite $N$:
\begin{equation}
J(N) =  a \exp\left(-b\, N^{0.5}\right)+ J(+\infty).
\end{equation} 
From this ansatz, we obtain the best fit $J(N)=1.054-6.665 \,\exp({-0.845\, N^{0.5}})$, from which we conclude that
\begin{equation}
J(+\infty)=1.054 \pm 0.013.
\end{equation}
The difference between the numerically estimated value and the predicted value of $J$ is thus less than the standard deviation, supporting our formula.

\section{Discussion}
\label{sec:discussion}

The main contribution of our manuscript is twofold. First, we have elucidated a general relationship between the modular commutator (Eq.~\eqref{eq:mod_comm_definition}), conditional mutual information, (bosonic) time reversal, modular flow. The modular commutator is a single quantity that possesses an intimate relationship with these seemingly unrelated objects, and this relationship calls for a further study of modular commutator in other physical systems.

Secondly, we provided nontrivial evidence for our formula for the chiral central charge, \emph{i.e.,} Eq.~\eqref{eq:main_result}. There were two main pieces of evidence, one being analytical and the other being numerical. The analytical evidence is the proof (in the framework of the entanglement bootstrap program~\cite{SKK2019,Shi2021}) that the formula yields the same value if two bulk phases are joined via a gapped domain wall. This result is consistent with the physical intuition that a gapped domain wall can exist only if the two bulk phases separated by the domain wall have the same chiral central charge. The numerical evidence was based on an exact calculation and an extrapolation based on model wave functions describing the semion model~\cite{Nielsen2012}. This model wave function has a chiral central charge of $c_-=1$, and our numerical study yields (upon extrapolation) a result consistent with $\approx 0.7\%$ error. These results add credence to the claim that the modular commutator is related to the chiral central charge for gapped quantum many-body systems in two spatial dimensions.

Below, we list natural questions that arise from our work, which we leave as open problems. 
\begin{enumerate}
\item Can one prove the quantization of the chiral central charge from the axioms of entanglement bootstrap~\cite{SKK2019}? In Section~\ref{sec:modular_current}, we only used one of the two entanglement bootstrap axioms (axiom {\bf A1}) and a physical argument to arrive at our formula. Perhaps, to make further progress on this problem, axiom {\bf A0} may play a role. 
\item Can one deduce symmetry-protected topological invaraints (\emph{e.g.,} ~\cite{Shiozaki2017,Shapourian2017,Shiozaki2018,Dehghani2021,Cian2021}) of the phase from the modular Hamiltonian? 
\item It will be interesting to calculate the modular commutator in other systems that have a chiral nature. Notable candidates are: (1) the critical point of the quantum phase transition between a pair of gapped chiral phases; (2) gapless systems (with a non-quantized) thermal Hall response, e.g., those proposed to be described by a chiral spin liquid with spinon Fermi surfaces \cite{Gong2019,Teng2020}; and (3) free-fermion systems, which may admit a simple form of the quantity \cite{Fidkowski2010}. 
Another interesting example is the Standard model of particle physics, which is a chiral theory; precisely speaking, it violates both parity $P$ \cite{wu1957experimental,lee1956question} and time reversal $T$ (which is better known as $CP$ violation) \cite{christenson1964evidence,cabibbo1963unitary,kobayashi1973cp}.
\item Is the modular commutator UV-finite in quantum field theory in general? In our manuscript, the UV-finiteness of the  modular commutator follows from the entanglement area law and its relation to the energy current at the edge. This line of reasoning will clearly not apply, for instance, if there is a correction to the entanglement entropy that scales logarithmically with the subsystem size~\cite{Calabrese2004,Fradkin2006,Metlitski2011,Casini2011}. 
\item It will be interesting to understand if the modular commutator has any bearing in quantum information theory. Clearly, the modular commutator quantifies something quantum; it is trivially zero for classical states because their reduced density matrices commute. However, it is unclear if the modular commutator has any operational interpretation in quantum information theory.
\item We showed that the modular commutator $J$ vanishes for states whose conditional mutual information vanishes on the three regions used to define $J$. We leave open the problem of proving a robust version of that statement that holds for states with nearly zero conditional mutual information, as is often the case for physical systems.
\item Chiral objects appear in nature in diverse forms, from large objects such as seashells and pinecones to their constituent molecules~\cite{cintas2013biochirality}. Chirality also has found its importance in synthetic chemicals, and the catalyzes that are needed to produce them~\cite{nguyen2006chiral}. These chiral molecules are examples of few-body quantum systems that possess chirality. (As objects with finite size, the thermal state can be a superposition of both chiralities, but the time it takes to relax to that equilibrium state can be longer compared with any practical time scales.) It is interesting to ask if the modular commutator or a suitable generalization can be useful to characterize the chirality of few-body quantum systems. 
\end{enumerate}

\subsection*{Acknowledgments:}
The authors thank Nikita Sopenko, Michael Gullans, and 
Maissam Barkeshli for valuable input and inspiring discussions. I.K. was supported by the Australian Research Council via the Centre of Excellence in Engineered Quantum Systems (EQUS) project number CE170100009. B.S. would like to thank Xuejun Guo for his intriguing lecture on Fibonacci numbers, the golden angle, and plants in the lecture series ``Ubiquitous Mathematics", the content of which can still be vividly recalled after years. B.S. further thank John McGreevy for kindly sharing the computer workstation in his office, from which the majority of data presented in this work is collected. B.S. is supported by the University of California Laboratory Fees Research Program, grant LFR-20-653926, and the Simons Collaboration on Ultra-Quantum Matter, grant 651440 from the Simons Foundation. K.K. is supported by MEXT Quantum Leap Flagship Program (MEXT Q-LEAP) Grant Number JP-MXS0120319794. Contributions to this work by NIST, an agency of the US government, are not subject to US copyright. Any mention of commercial products does not indicate endorsement by NIST. V.V.A. thanks Olga Albert, Halina and Ryhor Kandratsenia, as well as Tatyana and Thomas Albert for providing daycare support throughout this work.

\bibliography{bib}

\appendix

\section{Modular response of entanglement entropy}\label{sec:appSJ}

In this Section, we show that modular commutator is the response of the entanglement entropy under a modular flow. Specifically, 
\begin{equation}
	\frac{d}{ds} S(\rho_{AB}(s))|_{s=0} = J(A,B,C)_{\rho}. \label{eq:response_appendix}
\end{equation}
for $\rho_{ABC}(s) = e^{i K_{BC}s} \rho_{ABC} e^{-iK_{BC}s}.$ Below, we derive this identity for positive definite density matrix $\rho_{ABC}>0$.\footnote{Any density matrix can be made positive definite by adding a tiny perturbation. Specifically, we can take $\rho \to \rho(1-\epsilon) + \epsilon \sigma$, where $\sigma$ is a full-rank state,  and then take the $\epsilon \to 0$ limit.}
To derive this identity, note that, for general one-parameter family of density matrices $\rho(s)$, 
\begin{widetext}
	\begin{equation}
		\frac{d}{ds}S(\rho(s))|_{s=0} =-\text{Tr}\left(\frac{d\rho(s)}{ds} \ln \rho(s)\right)|_{s=0}  - \text{Tr}\left(\rho(s) \frac{d}{ds}\ln \rho(s)\right)|_{s=0}.
	\end{equation} 
\end{widetext}
By Duhamel's formula, (suppressing the explicit $s$-dependence)
\begin{equation}
	\begin{aligned}
		\frac{d\rho}{ds} &= \frac{d}{ds} e^{\ln \rho} \\
		&= \int_0^1 \rho^t \left(\frac{d}{ds}\ln \rho \right) \rho^{1-t}dt.
	\end{aligned}
\end{equation}
Dividing both sides by $\rho^{\frac{1}{2}}$, both on the left and the right, we obtain
\begin{equation}
	\begin{aligned}
		\rho^{-\frac{1}{2}} \frac{d \rho}{ds} \rho^{-\frac{1}{2}} &= \int_{-\frac{1}{2}}^{\frac{1}{2}} \rho^t \left(\frac{d}{ds} \ln \rho \right) \rho^{-t} dt \\
		&=\sum_{i, j}\left(\frac{d}{ds} \ln \rho\right)_{ij} |i\rangle\langle j| \frac{\sinh \left( \frac{\ln \lambda_i-\ln \lambda_j}{2}\right)}{\frac{\ln\lambda_i-\ln\lambda_j}{2}},
	\end{aligned}
\end{equation}
where $\{ |i\rangle \}$ and $\{ \lambda_i \}$ are the set of eigenstates and the eigenvalues of $\rho(s)$, respectively. Thus, we obtain
\begin{equation}
	\frac{d\ln \rho}{ds} = \Phi\left[ \rho^{-\frac{1}{2}} \frac{d\rho}{ds} \rho^{-\frac{1}{2}}\right],
\end{equation}
where $\Phi$ is a linear operator defined as 
\begin{equation}
	\Phi[O] = \sum_{i,j} O_{ij} |i\rangle \langle j| \frac{ \frac{\ln \lambda_i - \ln \lambda_j}{2}}{\sinh\left(\frac{\ln\lambda_i - \ln\lambda_j}{2} \right)}
\end{equation}
for any operator $O$. Moreover, the dual of $\Phi$ (denoted as $\Phi^{\dagger}$) with respect to the Hilbert-Schmidt inner product satisfies $\Phi^{\dagger}[\rho] =\rho$. Thus, 
\begin{equation}
	\begin{aligned}
		\text{Tr}\left(\rho \frac{d}{ds}\ln \rho \right) &= \text{Tr} \left(\rho \rho^{-\frac{1}{2}} \frac{d\rho}{ds} \rho^{-\frac{1}{2}}\right) \\
		&= \text{Tr}\left(\frac{d\rho}{ds}\right) \\
		&= 0.
	\end{aligned}
\end{equation}

Thus, we conclude that
\begin{equation}
	\frac{d}{ds}S(\rho)|_{s=0} = -\text{Tr}\left(\frac{d\rho}{ds} \ln \rho \right)|_{s=0}.
\end{equation}
Plugging in $\rho_{ABC}(s) = e^{i K_{BC}s} \rho_{ABC} e^{-iK_{BC}s}$, we obtain Eq.~\eqref{eq:response_appendix}.


Below is an alternative way to see $\text{Tr}\left(\rho \frac{d}{ds}\ln \rho \right) =0.$ First, this is true when $\rho(s)$ is classical, namely when $\rho(s)=\sum_i \lambda_i\vert i\rangle \langle i\vert$, where the orthonormal basis $\{\vert i\rangle \}$ is independent of $s$; this follows from $\text{Tr}\left(\rho \frac{d}{ds}\ln \rho \right)= \sum_i \lambda_i \frac{d}{ds} \ln \lambda_i= \sum_i \frac{d}{ds} \lambda_i = \frac{d}{ds}\text{Tr}\rho=0$.

For the quantum case, the basis can change as well. For small $s$, in a proper choice of basis\footnote{When the eigenvalues are nondegenerate, the basis is the diagonal basis. For degenerate case choose one that works.} we have
\begin{equation}
	\begin{aligned}
		\rho&=U\left(\sum_i \lambda_i \vert i\rangle \langle i\vert \right) U^{\dagger}.\\
		\ln\rho&=U\left(\sum_i \ln\lambda_i \vert i\rangle \langle i\vert \right) U^{\dagger}.\\	
	\end{aligned}	
\end{equation}
Here, the unitary operator $U$ and the eigenvalues $\{\lambda_i\}$ depend on $s$, whereas the orthonormal basis $\{ \vert i\rangle\}$ does not.
Therefore
\begin{equation}
	\begin{aligned}
		\text{Tr}\left(\rho \frac{d}{ds}\ln \rho \right) &= \text{Tr} \left(\rho\frac{dU}{ds}U^{\dagger}\ln\rho + \rho\ln\rho \, U\frac{dU^{\dagger}}{ds} \right) \\
		&+ \sum_i \lambda_i \frac{d\ln\lambda_i}{ds}\\
		&= 0 .
	\end{aligned}
\end{equation}
We have used  $\frac{dU}{ds}U^{\dagger}+U\frac{dU^{\dagger}}{ds}=0$, which follows from unitarity.

\section{Numerical Method}
In this Section, we discuss a numerical method to compute $J(A,B,C)$, given access to the full many-body wave function. Note that the modular commutator can be expressed as 
\begin{equation}
    J(A,B,C) = 2\text{Im}(\langle \widetilde{\psi}_{AB} |\widetilde{\psi}_{BC}\rangle),
\end{equation}
where $|\widetilde{\psi}_X\rangle = K_X |\psi\rangle$. Thus, calculation of $J(A,B,C)$ can be reduced to the calculation of $|\widetilde{\psi}_X\rangle$.

While this calculation can be performed by computing the reduced density matrix over $X$, taking the logarithm, tensor with the identity on the complement of $X$, and then applying this operator to the state $|\psi\rangle$, there is a more efficient approach. Given $|\psi\rangle$, we can always perform a Schmidt decomposition over the partition of the system into $X$ and its complement. Denoting the complement of $X$ as $\bar{X}$, we get
\begin{equation}
    |\psi\rangle_{X\bar{X}} = \sum_{n} \sqrt{p}_n |\phi_n\rangle_X |\psi_n\rangle_{\bar{X}},
\end{equation}
where $\{p_n: p_n\geq 0, \sum_n p_n=1 \}$ is a set of Schmidt values. Moreover, one can verify that
\begin{equation}
    |\widetilde{\psi}_X\rangle = - \sum_{n} \sqrt{p}_n \ln p_n |\phi_n\rangle_X |\psi_n\rangle_{\bar{X}}.
\end{equation}
Therefore, to compute $|\widetilde{\psi}_X\rangle$, it suffices to obtain a Schmidt decomposition of $|\psi\rangle_{X\bar{X}}$.

The Schmidt decomposition can be obtained by first mapping the $|\psi\rangle_{X\bar{X}}$ to a matrix and calculating the singular value decomposition of this matrix. Without loss of generality, suppose $|\psi\rangle_{X\bar{X}} = \sum_{i,j} \psi_{ij} |i\rangle_X|j\rangle_{\bar{X}}$, where $\{|i\rangle_X \}$ and $\{|j\rangle_{\bar{X}} \}$ are orthonormal basis sets for the Hilbert space over $X$ and $\bar{X}$ respectively. One can simply define a matrix $\sum_{i,j} \psi_{ij} |i\rangle_X\langle j|$ and compute its singular value decomposition. Now the calculation of $|\widetilde{\psi}_{X\bar{X}}\rangle$ is straigthforward. One can simply change each singular value $\lambda $ to $-2\lambda \ln \lambda$. Mapping the resulting operator back to a state, we obtain $|\widetilde{\psi}_{X}\rangle$.

\end{document}